%% file: main.tex
\documentclass[letterpaper, 11pt]{article}
\usepackage[utf8]{inputenc}
\usepackage[margin=1in]{geometry}
\usepackage{amsthm}
\usepackage{thm-restate}
\usepackage{xspace}
\usepackage{amssymb}
\usepackage{amsmath}
\usepackage{amsfonts}
\usepackage{bbm}
\usepackage{xcolor} 
\usepackage{comment}
\usepackage{booktabs}
\usepackage{graphicx}
\usepackage{diagbox}
\usepackage{amsmath,amsfonts,graphicx,amssymb,bm,amsthm}
\usepackage{algorithm,algorithmicx}
\usepackage[noend]{algpseudocode}
\usepackage{fancyhdr}
\usepackage{tikz}
\usepackage{graphicx}
\usepackage{hyperref}
    \hypersetup{ colorlinks=true, linkcolor=blue, filecolor=magenta, urlcolor=red,}
\usetikzlibrary{arrows,automata}
\usepackage[capitalize]{cleveref}
\usepackage{mathabx}

\usepackage{multirow,array}
\newtheorem{lemma}{Lemma}
\newtheorem{theorem}{Theorem}
\newtheorem{corollary}{Corollary}
\newtheorem{definition}{Definition}
\newtheorem{conjecture}{Conjecture}

\newtheorem{property}{Property}
\newtheorem{claim}{Claim}

\newcommand{\C}{\mathbb C}
\newcommand{\Cout}{\mathbb C_{\textsf{out}}}
\newcommand{\Cin}{\mathbb C_{\textsf{in}}}
\newcommand{\corrupted}{ y}

\newcommand{\alphabetout}{\Sigma_{\text{out}}}
\newcommand{\alphabetin}{\Sigma_{\text{in}}}
\newcommand{\ed}{\Delta_E}
\newcommand{\hamming}{\Delta_H}
\newcommand{\eps}{\varepsilon}
\newcommand{\poly}{\mathsf{poly}}
\newcommand{\polylog}{\mathsf{polylog}}
\newcommand{\lcs}{\mathsf{LCS}}
\newcommand{\enc}{\mathsf{Enc}}
\newcommand{\F}{\mathsf{F}}

\newcommand{\ent}{\mathsf{H}}

\newcommand{\authnote}[3]{\textcolor{#2}{{\sf (#1's Note: {\sl{#3}})}}}
\newcommand{\knote}{\authnote{Kuan}{cyan}}
\newcommand{\xnote}{\authnote{Xin}{blue}}

\begin{document} 

\title{Linear Insertion Deletion Codes in the High-Noise and High-Rate
Regimes}
\author{
Kuan Cheng \thanks{Department of Computer Science, Peking University, \texttt{ckkcdh@pku.edu.cn}. Supported by a start-up fund of Peking University.} 
\and Zhengzhong Jin \thanks{Massachusetts Institute of Technology, \texttt{zzjin@mit.edu}. Supported in part by DARPA under Agreement No. HR00112020023 and by an NSF grant CNS-2154149. Most of the work was done while the author was a PhD student at Johns Hopkins University, and supported by NSF CAREER Award CCF-1845349.}
\and Xin Li \thanks{Department of Computer Science, Johns Hopkins University, \texttt{lixints@cs.jhu.edu}. Supported by NSF CAREER Award CCF-1845349 and NSF Award CCF-2127575.}
\and Zhide Wei \thanks{Department of Computer Science, Peking University, \texttt{zhidewei@pku.edu.cn}.  Supported by a start-up fund of Peking University.}
\and Yu Zheng \thanks{Meta Platforms Inc, \texttt{hizzy1027@gmail.com}. Most of the work was done while the author was a PhD student at Johns Hopkins University, and partially supported by NSF Award CCF-2127575.}.
}
\date{}

\maketitle

\input{abstract.tex}
\input{intro}
\input{organization.tex}
\input{prelim.tex}

\input{general_construction.tex}

\input{construction_herror}
\input{construction_hrate1}

\input{construction_hrate2}
\input{existence.tex}

\input{half_plus_logn_lower_bound}

\section{Conclusion and Open Problems} \label{sec:open}
In this paper we give explicit constructions of linear insdel codes that can either achieve close to optimal rate or close to optimal faction of correctable errors. A natural open question is to improve the trade-off of the parameters in our constructions. For example, one ultimate goal would be to construct linear insdel codes that approach the half-singleton bound over some constant size alphabet. Another open question is to prove or refute our conjecture. Finally, it would also be interesting to determine the best possible fraction of errors a binary (or over any alphabet $\F_q$) linear insdel code can correct while still mainating a positive rate.

\bibliographystyle{alpha} 
\bibliography{ref}

\end{document}

%% file: abstract.tex
\begin{abstract}
    This work continues the study of linear error correcting codes against adversarial insertion deletion errors (insdel errors). Previously, the work of Cheng, Guruswami, Haeupler, and Li \cite{CGHL21} showed the existence of asymptotically good linear insdel codes that can correct arbitrarily close to $1$ fraction of errors over some constant size alphabet, or achieve rate arbitrarily close to $1/2$ even over the binary alphabet. As shown in \cite{CGHL21}, these bounds are also the best possible. However, known explicit constructions in \cite{CGHL21}, and subsequent improved constructions by Con, Shpilka, and Tamo \cite{9770830} all fall short of meeting these bounds. Over any constant size alphabet, they can only achieve rate $< 1/8$ or correct $< 1/4$ fraction of errors; over the binary alphabet, they can only achieve rate $< 1/1216$ or correct $< 1/54$ fraction of errors. Apparently, previous techniques face inherent barriers to achieve rate better than $1/4$ or correct more than $1/2$ fraction of errors.

    In this work we give new constructions of such codes that meet these bounds, namely, asymptotically good linear insdel codes that can correct arbitrarily close to $1$ fraction of errors over some constant size alphabet, and binary asymptotically good linear insdel codes that can achieve rate arbitrarily close to $1/2$.\ All our constructions are efficiently encodable and decodable. Our constructions are based on a novel approach of code concatenation, which embeds the index information implicitly into codewords. This significantly differs from previous techniques and may be of independent interest. Finally, we also prove the existence of linear concatenated insdel codes with parameters that match random linear codes, and propose a conjecture about linear insdel codes. 
\end{abstract}

%% file: intro.tex
\section{Introduction}
Error correcting codes are fundamental objects in computer science and information theory. Starting from the seminal work of Shannon and Hamming, the study of error correcting codes has led to a deep understanding of how to ensure reliable communications in various noisy channels. Furthermore, error correcting codes have found rich applications in other seemingly unrelated areas such as complexity theory, learning theory, pseudorandomness and many more. Traditionally, the errors studied are either erasures (where a transmitted symbol is replaced by a `?') or symbol modifications (where a transmitted symbol is replaced by a different symbol), and they can be either random or adversarial. Through decades of effort, we now have an almost complete understanding of codes for such errors, and constructions with efficient encoding and decoding algorithms that match or are close to various known bounds. 

An important and more general type of errors, known as \emph{synchronization errors}, however, is much less understood. These errors include insertion and deletions (so we also call them insdel errors for short), which can cause the positions of received symbols to shift. On the other hand, they occur frequently in real world applications, including disk access, integrated circuits, communication networks and so on. They are also closely related to applications in computational biology and DNA-based storage systems \cite{Bornholt16, Olgica17}. Although the study of codes for such errors started around the same time as Shannon's works, progress has historically been slow due to the apparent difficulty of handling the loss of index information with such errors. For example, many basic questions, such as the capacity of the binary deletion channel with
deletion probability $p$ is still wide open, and the first explicit construction that has a constant rate and can correct a constant fraction of adversarial errors is not known until 1999 \cite{SZ99}.

From now on, we will focus exclusively on adversarial insdel errors. Over the past several years, with the development of new techniques such as \emph{synchronization strings} \cite{HaeuplerS17}, there has been a wave of new constructions of codes for these errors \cite{HaeuplerS17, HRS19LinearTimeInsdel, SimaB19, cheng2018deterministic, BGZ-journal, Haeupler19, SychnStringsInteractive,GW, BGH17, cheng_et_al:LIPIcs:2019:10613, HSS18ListInsdel, GHS-stoc20,hayashi2018list,wachter2017list,GW}. Some of them achieve excellent parameters, e.g., codes that approach the singleton bound over a large constant alphabet \cite{HaeuplerS17}, codes with almost optimal redundancy over the binary alphabet \cite{cheng2018deterministic, Haeupler19},  list-decodable codes over large alphabets that can correct more errors than the length of the codeword \cite{HSS18ListInsdel}, and list-decodable codes over any alphabet of positive rate for the information-theoretically largest possible combination of insertions and deletions \cite{GHS-stoc20,hayashi2018list,wachter2017list,GW}. However, none of the above constructions gives a \emph{linear} code, and the existence of asymptotically good linear codes for insdel errors over a constant size alphabet is not known until the work of Cheng, Guruswami, Haeupler, and Li \cite{CGHL21}.

The motivation of studying linear codes comes from several aspects. First, they have compact representations using either generator matrices or parity check matrices, which directly give efficient encoding and testing algorithms with running time $O(n^2)$. Second, such codes have simple structures, so they are often easier to analyze and allow one to use powerful techniques from linear algebra. Finally, linear codes have had great success in codes for erasures and symbol modifications, achieving some of the most well known constructions with (near) optimal parameters. Thus, one could ask if the same is true for insdel codes.

As is standard in the literature of error correcting codes, the two most important parameters of a linear insdel code are $\delta$, the fraction of insdel errors the code can correct; and $R$, the rate of the code, defined as the message length divided by the codeword length. In \cite{CGHL21}, the authors established several bounds regarding the tradeoff between these two parameters for linear insdel codes. First, they showed that any linear code correcting $\delta$ fraction of insdel errors must have rate at most $\frac{1}{2}(1 - \delta)$, regardless of the alphabet size. This is known as the \emph{half-singleton bound} and generalizes a previous result in \cite{AGFC07}, which shows that any linear code that can correct even a single deletion must have a rate of at most $1/2$. This bound shows a severe limitation of linear codes for insdel errors, as general codes can correct $\delta$ fraction of errors with $R$ approaching $1-\delta$. Taking into consideration the alphabet size $q$, this bound can be improved to $\frac{1}{2}(1 - \frac{q}{q-1}\delta) + o(1)$, which is known as the \emph{half-Plotkin bound}. On the other hand, the authors also showed that over the field $\F_q$, for any $\delta>0$ there exists a linear code family that can correct $\delta$ fraction of insdel errors, with rate $(1-\delta)/2-\ent(\delta) / \log_2 q$, where $\ent$ is the binary entropy function. In particular, this implies the existence of binary linear codes with rate $1/2 - \eps$ capable of correcting $\Omega(\eps \log^{-1} \frac{1}{\eps})$ fraction of insdel errors for any $\eps>0$; and linear insdel codes over $\F_q$ of rate $\frac{1}{2}(1 - \delta) - \eps $ capable of correcting any $\delta$-fraction of insdel errors, for a large enough $q =2^{ \Theta(\eps^{-1}) }$, which approaches the half-singleton bound. Hence, the rate can approach $1/2$ even over the binary alphabet, and the fraction of errors corrected can approach $1$ over a constant size alphabet, both of which are the best possible.

Going further, \cite{CGHL21} also constructed explicit asymptotically good linear insdel codes. However, the fraction of errors the code can correct and the rate of the code are both quite small. \cite{CGHL21} did not specify these constants, but a rough estimate shows that the code has $\delta< 1/400$ and $R < 2^{-80}$. Thus a natural question left in their work is to improve these parameters.

Recently, a subsequent work by Con, Shpilka, and Tamo \cite{9770830} made progress in this direction.\ For a field $\F_q$ with $q=\poly(1/\eps)$, they constructed explicit linear insdel codes that can correct $\delta$ fraction of errors with rate $R=(1-4 \delta)/8-\eps$. For the field $\F_2$ their explicit linear code can correct $\delta$ fraction of errors with rate $R=(1-54 \delta)/1216$. Hence, for a constant size alphabet their construction can achieve $\delta < 1/4$ with a positive $R$, or $R < 1/8$ with a positive $\delta$. For the binary alphabet, their construction can achieve $\delta < 1/54$ with a positive $R$, or $R < 1/1216$ with a positive $\delta$. One caveat is that their codes over the binary alphabet can only decode efficiently from deletions (although they can also decode from insertions inefficiently), while their codes over the large alphabet can decode efficiently from both deletions and insertions. In another work by the same authors \cite{9834672}, they also showed the existence of Reed-Solomon codes over a field of size $n^{O(k)}$ that have message length $k$, codeword length $n$, and can correct $n-2k+1$ insdel errors. This achieves the half-singleton bound. They complemented the existential result by providing a deterministic construction over a field of size $n^{k^{O(k)}}$, which runs in polynomial time for $k=O(\log n/\log \log n)$. Nevertheless, in this paper we only focus on the case of a constant alphabet size. 

In summary, all known explicit constructions over constant size alphabets fall short of getting rate close to $1/2$, or getting the fraction of errors correctable close to $1$. In fact, previous techniques seem to face inherent barriers to achieve rate better than $1/4$ or correct more than $1/2$ fraction of errors, which we will talk about in more details when we give an overview of our techniques.

\subsection{Our Results}
In this paper we further improve the fraction of errors $\delta$ and the rate $R$ that can be achieved by linear insdel codes with efficient encoding and decoding algorithms. In the case of high noise, we give explicit constructions of insdel codes with positive rate that can correct $\delta$ fraction of errors with $\delta$ arbitrarily close to $1$, over a constant size alphabet. In the case of high rate, we give explicit constructions of insdel codes that can achieve rate arbitrarily close to $1/2$ and correct a positive constant fraction of errors, over the binary alphabet.\footnote{It's also easy to generalize our constructions to larger alphabet size, but for clarity we omit the details in this version.} Specifically, we have the following theorems.

\begin{theorem}[High noise]
For any constant $\eps>0$ there exists an efficient construction of linear insdel codes over an alphabet of size $\poly(1/\eps)$, with rate $\Omega(\eps^2)$ that can correct $1-\eps$ fraction of insdel errors (possibly inefficiently).
\end{theorem}

With efficient decoding, the rate becomes slightly worse.

\begin{theorem}[High noise]
    For any constant $\eps >0$, there is a family of linear codes with rate $\Omega(\eps^4)$ and alphabet size $\poly(1/\eps)$, that can be encoded in polynomial time and decoded from up to $1-\eps$ fraction of insdel errors in polynomial time.
\end{theorem}

\begin{theorem}[High rate]
For any constant $\eps > 0$, there is a family of binary linear codes with rate $1/2 - \eps$, that can be encoded in polynomial time and decoded from $\Omega(\eps^3 \log^{-1}\frac{1}{\eps})$ fraction of insdel errors in polynomial time. 
\end{theorem}

Our constructions are based on code concatenation. We complement our explicit constructions by showing that there exist linear concatenated codes that match the parameters of random linear codes. These constructions can be considered in a sense ``semi-explicit" since the outer code is explicit, and we only need to find explicit inner codes.

\begin{theorem}\label{thm:existence}
For any field $\F_{q_0}$ and any constant $\delta>0$, there exists a family of linear concatenated code over $\F_{q_0}$ where the outer code is a Reed-Solomon code, such that the code has rate $\frac{1}{2}(1- \delta)- \ent(\delta)/\log q_0-o(1)$ and can correct $\delta$ fraction of insdel errors, where $\ent()$ is the binary entropy function. 
\end{theorem}
We emphasize that the inner codes here may be different for different positions.
So if one wants to use brute force to search for a sequence of proper inner codes, then this may take time at least $2^{n \log^2 n}$ where $n$ is the length of the outer codewords.

This theorem implies the following corollaries.

\begin{corollary}
For any constant $\delta > 0$, there exists a family of binary linear concatenated code where the outer code is a Reed-Solomon code, such that the code has rate $\frac{1}{2}(1- \delta) $ and can correct $\Omega(\delta\log^{-1}\frac{1}{\delta})$ fraction of insdel errors.
\end{corollary}

\begin{corollary}
For any constants $\delta, \eps>0$ there exists a family of linear concatenated code over an alphabet of size $q =2^{ \Theta(\eps^{-1}) }$ where the outer code is a Reed-Solomon code, such that the code has rate $\frac{1}{2}(1- \delta)-\epsilon$ and can correct $\delta$ fraction of insdel errors.
\end{corollary}

Finally, we study the question of whether binary linear insdel codes can achieve $\delta$ arbitrarily close to $1/2$ with a positive rate $R$. Notice that even for general binary codes, it is well known that the maximum fraction of deletions that any non-trivial binary code of size $\ge 3$ can correct is below $1/2$ since any $3$ different $n$-bit binary strings must contain two strings with the same majority bit, and thus their longest common subsequence is at least $n/2$. For binary linear codes this can also be seen from the half-Plokin bound. A recent work by Guruswami, He, and Li \cite{9954058} in fact already provided a negative answer to this question even for general binary codes. In particular, they showed that there exists an absolute constant $\alpha>0$ such that any binary code $\C \subseteq \{0, 1\}^n$ with $|\C| \geq 2^{\polylog n}$ must have two strings whose longest common subsequence has length at least $(1/2+\alpha)n$. Thus $\C$ cannot correct more than $1/2-\alpha$ fraction of insdel errors. Since linear codes are more restricted, one may expect that a stronger result can be proved for binary linear codes. Specifically, we have the following conjecture:

\begin{conjecture}
There exists an absolute constant $\alpha>0$ such that any linear subspace $\C \subseteq \F_2^n$ with dimension $\geq 3$ must have two strings (vectors) whose longest common subsequence has length at least $(1/2+\alpha)n$.
\end{conjecture}

However, we are not able to prove this conjecture. Instead, we can prove a weaker result.

\begin{theorem}
There exists an absolute constant $\alpha>0$ such that any linear subspace $\C \subseteq \F_2^n$ with dimension $\geq 3$ must have two strings (vectors) whose longest common subsequence has length at least $(\frac{1}{2}+\frac{\alpha}{\log n})n$.
\end{theorem}

\input{overview.tex}

%% file: overview.tex
\subsection{Overview of the Techniques}
There have been only two previous works on explicit constructions of asymptotically good linear insdel codes over fields of constant size, i.e., \cite{CGHL21} and \cite{9770830}. The apparent difficulty of constructing such codes comes from the following aspects: First, many of the previous constructions of (non-linear) insdel codes are based on adding index information to the codewords, either in the form of direct encoding of indices, or more sophisticated objects such as synchronization strings. Since all of these result in fixed strings, adding such information in any naive way will lead to non-linear codes. Indeed, both \cite{CGHL21} and \cite{9770830} have to find alternative ways to ``embed" synchronization strings into a linear code. Specifically, \cite{CGHL21} uses what is called a \emph{synchronization sequence}, which is a sequence of $0$'s added in between each pair of adjacent symbols in a codeword. This preserves the linearity if the original code is linear. \cite{9770830}, on the other hand, embeds the synchronization string by combining a codeword symbol $x$ and a synchronization string symbol $a$ into a pair $(x, a \cdot x)$, where $\cdot$ is the multiplication over the corresponding field $\F_q$. This also preserves the linearity over $\F_q$, but now the symbols from the synchronization strings are mixed with symbols from the codeword, and it is not easy to tell them apart. Note that for decoding, one needs to first use the synchronization string to recover the positions of the codeword symbols. To solve this problem, \cite{9770830} also needs to add buffers of $0$'s between adjacent pairs, where the length of a buffer is at least as long as the pair $(x, a \cdot x)$. 

It can be seen that the added $0$'s in the above two approaches form an inherent barrier to achieving high rate or high fraction of correctable errors. In \cite{CGHL21}, a constant number of $0$'s are added in between each pair of adjacent symbols in a codeword, which already decreases the rate and the possible decoding radius to a small constant. In \cite{9770830}, the operation of converting a codeword symbol $x$ and a synchronization string symbol $a$ into a pair $(x, a \cdot x)$ already decreases the rate of the code to below $1/2$, while adding $0$'s as buffers decreases the rate even more to below $1/4$. Similarly, add $0$'s as buffers also decreases the possible decoding radius to below $1/2$. For binary codes, \cite{9770830} needs to use another layer of code concatenation, which further decreases the rate and decoding radius.

The key idea in all our constructions is to eliminate the use of $0$'s as buffers or synchronization sequences. Instead, we embed synchronization information directly into the codewords. To achieve this, we also use code concatenation, where for the outer code we choose a suitable Reed-Solomon code. On the other hand, the key difference between our constructions and standard concatenated codes is that we choose a \emph{different} inner code for \emph{every position} of the outer code. This way, we can make sure that the inner codewords corresponding to outer codeword symbols at different positions are far enough from each other, and thus we can roughly tell them apart by just looking at the received codeword. By using linear inner codes for all positions, this preserves the linearity of the code, and at the same time eliminates the use of $0$'s. On a high level, this is why our constructions can achieve either high rate (arbitrarily close to $1/2$) or high fraction of correctable errors (arbitrarily close to $1$). We now discuss our techniques in more details for the two cases.

\subparagraph{Constructions for high error.} Note that to correct $1-\eps$ fraction of insdel errors, a linear code must have alphbet size at least $1/\eps$ by the half-Plotkin bound. Here we use an alphabet of size $\poly(1/\eps)$. With an appropriately chosen parameter $\gamma=\Omega(\eps)$, after picking an outer Reed-Solomon code with codeword length $n$, rate $\gamma$ and relative distance $1-\gamma$, our strategy is to design $n$ different inner codes $\Cin^1, \cdots, \Cin^n$. The goal is to ensure that codewords in different inner codes have large edit distance, or equivalently, the length of their longest common subsequence (LCS for short) is at most $\gamma n'$ where $n'=O(\log n)$ is the block length of the inner code. However, since all these codes are linear, $0$ is a codeword of each inner code, and two $0$'s (even from different inner codes) are guaranteed to have $0$ edit distance. We design the inner codes to ensure this is the only bad case. 

More specifically, we ensure that for any two inner codewords $x, y$, unless they are both $0$ or they correspond to the same message in one inner code $\Cin^i$, their edit distance is large. We show that if we pick $n$ random linear codes for $\Cin^1, \cdots, \Cin^n$, then this property holds with high probability. Furthermore, we can derandomize this by using a small biased sample space to generate the $n$ generator matrices of $\Cin^1, \cdots, \Cin^n$. Roughly, this is because the property we want is \emph{local} --- it only looks at any two inner codewords  $x, y$. By using a small biased sample space, we can show that (roughly) under the above conditions, any non-trivial parity of the bits (we treat a symbol in the alphabet of size $\poly(1/\eps)$ as a binary string of length $O(\log(1/\eps))$) of $(x, y)$ has a small bias. Hence a standard XOR lemma implies the joint distribution of $(x, y)$ is close to uniform. Since $n'=O(\log n)$, we only need to look at $\poly(n)$ such pairs of $(x, y)$. Thus it suffices to choose the error in the small biased sample space to be $1/\poly(n)$. This gives us a sample space of size $\poly(n)$ and we can exhaustively search for a good construction. This gives us $n$ different inner codes with rate $\Omega(\gamma)$.

Using these inner codes, it is now relatively straightforward to argue about the parameters of the concatenated code. The rate is $\Omega(\gamma^2)=\Omega(\eps^2)$. To argue about the distance, we consider the LCS between any two different codewords $C_1, C_2$, and divide it sequentially into blocks according to the inner codewords of $C_1$. Each block now covers a substring of $C_2$. Intuitively, by the property of our inner codes, each block contains only a small number of matches compared to the total size of this block in $C_1$ and the substring covered in $C_2$, unless it is a $0$ inner codeword in $C_1$ and is matched to another $0$ inner codeword in $C_2$, or it is a match between the same inner codeword in a single inner code $\Cin^i$. However our outer code guarantees that the latter cannot happen too many times (i.e., at most $O(\gamma n)$ times). Therefore the LCS has length at most $O(\gamma nn')$. By choosing $\gamma$ appropriately, the code can correct $1-\eps$ fraction of insdel errors.

We present a simple polynomial time decoding algorithm. Given any received string $y$, we consider the partition of $y$ into $n$ substrings $y_1, \cdots, y_n$ such that $y =y_1 \circ y_2 \circ \cdots \circ y_n$, where each $y_i$ can be the empty string. For each $y_i$, we find the closest codeword $x_i \in \Cin^i$ in edit distance and record their edit distance $\Delta_i$. We then minimize $\Delta=\sum_{i \in [n]} \Delta_i$, by using a simple dynamic programming. We show that as long as there are not too many errors, by using the optimal partition returned from the dynamic programming, one can correctly recover a small fraction of the outer codewords. Intuitively, this is because if the partition results in too many errors in the recovered outer codewords, then again by the property of our inner codes, the quantity $\Delta$ will be very large, unless there are a lot of errors. We then use a list decoding algorithm for the Reed-Solomon code to get a list of candidate codewords, and search the list to find the correct codeword, which is the one closest to $y$ in edit distance.\ For technical reasons, this decreases the rate of the code to $\Omega(\eps^4)$.

\subparagraph{Constructions with high rate.}
In section~\ref{sec:high rate 1} and section~\ref{sec:high rate 2} we show our constructions with high rate and polynomial time encoding and decoding.
Section~\ref{sec:high rate 1} exhibits a warm-up construction achieving rate $1/3-\gamma$,
while section~\ref{sec:high rate 2} improves it to $1/2-\gamma$. These constructions are significantly more involved due to additional issues arised in the analysis.

The constructions in both sections inherit the structure of the general construction.
The outer code is a Reed-Solomon code with block length $n$, alphabet size $n$, relative distance $\delta$ and rate $(1-\delta)$, where $\delta=\gamma/2$. To achieve a high rate, we will design the inner codes to have a large rate, ideally close to $1/2$. At the same time, we also need to ensure that the code can correct a positive constant fraction of errors, thus we want to make sure that the LCS between any two different codewords is not too large.

As before, we will design the inner codes such that ideally, codewords from different inner codes are far away from each other (or equivalently, have small LCS). However, there are additional issues in the analysis of the LCS. First, the $0$ codewords from different inner codes are always the same. This is inevitable since we are dealing with linear codes. Second, in a matching between two different codewords $C_1, C_2$, some inner codeword of $C_1$ may be matched to a substring of the concatenation of two adjacent inner codewords of $C_2$. Thus it is not enough to just ensure that codewords from different inner codes are far away from each other. We note that this issue also occurs in our constructions for high noise. However, there we designed the inner codes to have small rate but large distance, so the LCS between different inner codewords is quite small. When some inner codeword of $C_1$ is matched to a substring of the concatenation of two adjacent codewords of $C_2$, the size of the matching in this part at most doubles the size of the LCS between two different inner codewords, and is affordable in that case. Here however, since we are trying to achieve a high rate, the distance between two different inner codewords becomes quite small, and the LCS becomes relatively large (e.g., larger than $1/2$ fraction).\ Hence, we cannot afford to double this size.

On a high level, we resolve the second issue by strengthening our local property of inner codes, while our analysis will show that the first issue can also be resolved as a consequence. We begin by discussing the local property we need in Section~\ref{sec:high rate 1} to achieve rate $1/3-\gamma$. The distinct binary inner codes $\Cin^1, \Cin^2, \ldots, \Cin^n$ are constructed to have block length $n'$, message length $k'=(1/3-\gamma/2)n'$, with the following property: for every $i, j\in [n]$, for every codeword $w$ in $\Cin^i$, for every two codewords $u, v$ from two adjacent inner codes $\Cin^j, \Cin^{j+1}$, unless $w = u$ or $w=v$, the distance between $w$ and any substring of  $u\circ v$ is at least $d'=  \Omega(n')$.
We first explain why this property implies a good decoding radius and then explain how to construct these inner codes.

We show the decoding radius by directly providing the following decoding algorithm.
On an input $y$ which is a corrupted version of a codeword $z$, the algorithm first finds a string $\tilde{z} \in \{0, 1\}^{nn'}$ which has a maximum block matching with $y$.
A block matching is defined to be a set of matches where each match, denoted as $(i, [\alpha, \beta])$, consists of a non-zero inner codeword  $u \in \Cin^i$ and a substring $y_{[\alpha, \beta]}$, such that their edit distance is at most $d'/2$. Furthermore, the matching is monotone in the sense that the substrings of $y$ involved in the matching do not overlap and the matches cannot cross. 
We call $u$  a candidate string for the $i$-th block. We give a simple dynamic programming to find a maximum   block matching together with a corresponding sequence of candidates.
To construct $\tilde{z}$, we first fill these candidates to their corresponding blocks and then set all the other blocks to be $0$.

Now we show that as long as there are at most $\rho nn'$ errors for some small constant $\rho > 0$, $\tilde{z}$ agrees with $z$ in most of the blocks (inner codewords). To show this, divide $z$ into blocks $z^1 \circ z^2 \circ \cdots \circ z^n$ such that each $z^i$ corresponds to an inner codeword. Similarly, divide $y$ into blocks $y^1 \circ y^2 \circ \cdots \circ y^n$ such that each $y^i$ is the corrupted version of $z^i$.
Notice that there can be at most $ \frac{\rho nn'}{ d'/2} =  (c \gamma) n $ blocks with at least $d'/2$ errors, for some constant $c=c(\rho)$.
So the maximum block matching has size at least $ \hat{n} - c\gamma n$ where $\hat{n}$ is the number of non-zero blocks in $z$.
Now consider a maximum block matching and the sequence of candidates returned by the algorithm.
We show that there are at most $c \gamma n$ candidates that are not equal to the corresponding blocks of $z$, by using the local property.
As we fill all the other blocks to be $0$, this also implies there are at most $c \gamma n$ zero-blocks being incorrectly recovered.
Hence the algorithm correctly recovers $1-O(c\gamma)$ fraction of blocks in $z$.
By taking $c$ (and thus also $\rho$) to be a small enough constant, one can use the list-decoding algorithm of Reed-Solomon codes to recover $z$.

Next we explain how to construct the inner codes. We start by considering a random construction, that is, all the inner codes are independent random linear codes.
We show the local property holds with high probability.
Consider arbitrary codewords $w \in \Cin^i\setminus \{0\}, u \in\Cin^j, v\in \Cin^{j+1}$ for some $i,j\in [n]$, where $w \neq u$ and $w\neq v$.  
Here the inequality means the two codewords are either from different inner codes or they correspond to different messages in one inner code. 
Suppose there is a substring $w'$ of $u\circ v$, which has distance $<d'$ to $w$.
So the LCS between $w$ and $w'$ should be $\ell \ge \frac{|w|+|w'| - d'}{2}$.
Notice that $\ell \le |w| \le n'$.
Consider any monotone alignment between $w$ and $w'$.
Because $w \neq u, w\neq v$ and the inner codes are all independent and generated randomly, by a similar argument as in \cite{CGHL21}, the event that the alignment is indeed a matching of bits happens with probability at most $2^{-\ell}$.
We then apply a union bound over all possible alignments of size $\ell$ and all possible codewords $w, u, v$. A key observation is that the number of all possible codewords $w, u, v$ is $2^{3k'}$ since we have three different codewords here. However, we have $\ell \leq n'$. Therefore for the union bound to work, we have to set $k' < n'/3$. This is the reason that we can only achieve rate close to $1/3$ with this construction. 

Next, we derandomize the construction by replacing the uniform randomness used with an $\eps$-biased distribution. Here, as before, we crucially use the fact that our property for the inner codes is local: the only place where we use randomness is when we bound the probability that an alignment is a valid matching, and it only involves three codewords. Since $n' = O(\log n)$, by using a standard XOR Lemma and taking $\eps = 1/\poly(n)$, we can argue that when restricted to any three codewords, the $\eps$-biased distribution is $1/\poly(n)$ close to the uniform distribution in statistical distance. This is enough for the union bound since there are at most $\poly(n)$ such triples $w, u, v$.  

Since we only need $O(\log n)$ random bits to generate the above $\eps$-biased distribution, one can exhaustively search for a good construction that satisfies our local property. This also takes polynomial time since one only needs to check every triple of inner codewords.

In section~\ref{sec:high rate 2}, we add new ideas to bypass the rate $1/3$ barrier in the above construction, by giving a new local property of the inner codes.
Recall that the reason we need to choose $k' < n'/3$ in the above construction is that the alignment we consider in the local property consists of matches that involve three different codewords, which results in a $2^{3k'}$ term in the union bound, but the alignment has size at most $n'$. In the new local property, we generalize this by considering alignments that involve $2s+1$ different codewords for some integer $s$. In a simplified version, consider any two different codewords $C_1, C_2$ of the concatenated codes and an LCS between them, we analyze any $s$ consecutive inner codewords in $C_1$, and how they can be matched to a substring in $C_2$. Note that the $s$ consecutive inner codewords cannot be matched to a substring with length much larger than $sn'$, or there are already many unmatched bits in $C_2$. So we can imagine a new local property like the following: let $w$ be the concatenation of any $s$ adjacent inner codewords, and $u$ be the concatenation of any $s+1$ adjacent inner codewords. As long as the codewords in $w$ and $u$ are sufficiently different, the distance between $w$ and any substring of $u$ is at least $\Omega(n')$. The idea is that an alignment between $w$ and $u$ can have size up to $sn'$, while the union bound gives a $2^{(2s+1)k'}$ term. Thus we can potentially achieve $k' < \frac{s}{2s+1}n'$, and if $s$ is large enough, the rate is close to $1/2$. Note that the construction in Section~\ref{sec:high rate 1} corresponds to the case of $s=1$.

However, it is not straightforward to make this idea work. The main issue is that unlike the simple case of $s=1$, when we consider $s$ consecutive inner codewords for $s>1$, there can be multiple $0$ codewords in them, which can potentially be matched to the $0$ codewords in $u$. Furthermore, there can be inner codewords in $w$ and $u$ that correspond to the same message in a single inner code. These issues will increase the probability that the alignment is a valid matching and can cause the union bound to fail. To fix this, we require the ``unique" blocks in $w$ to be dense. Specifically, we define a unique block of $w$ (or $u$) to be a non-zero inner codeword such that either no block of $u$ (or $w$) is in the same inner code with it, or any block of $u$ (or $w$) in the same inner code with it corresponds to a different message.  Now we define the following new local property: 

For every $w$ which is a sequence of $t = O(\frac{\log \frac{1}{\gamma}}{\gamma^2})$ consecutive inner codewords, every $u$ which is a sequence of $t+1$ consecutive inner codewords, and every $w'$ which is a substring of $u$, the distance between $w$ and $w'$ is at least $d' = \Omega(\gamma n')$, as long as the number of unique blocks in $w$ or $u$ is at least $s = \Omega(\gamma t)$. By the distance property of the outer code, for any two different concatenated codewords, in at least one of them, the fraction of such $t$ consecutive inner codewords with at least $s$ unique blocks is a constant.

Using this new property, we can design a similar decoding algorithm as that of Section~\ref{sec:high rate 1}, and  with a similar analysis, achieve   decoding radius $ \Omega(\gamma^3 n/\log \frac{1}{\gamma}) $.

We defer these details to the technical part, and mainly explain here how to construct the inner codes with the new property and why this indeed gives a rate of $1/2-\gamma$.

Similar as before, we start with a construction where all inner codes are independent random linear codes, and later derandomize it with an $\eps$-biased space. As long as the parameters $s, t$ are constants, it is easy to see that the derandomization step still works. Therefore, now we only focus on the random construction and argue that the new local property holds with high probability. For this we use a delicate combinatorial and probabilistic argument.

Suppose the property is not satisfied with some concatenated codewords $C_1, C_2$. Then there exists a $ w'$ such that the edit distance between $w, w'$ is less than $d'$, which implies the LCS between $w$ and $w'$ is   $\ell > (|w| + |w'| - d')/2 $.
Consider an arbitrary monotone alignment $M$ between $w$ and $w'$ of size $\ell$.
We have two cases.
The first case is that there is a pair of indices $(i, j)$ in $M$ such that $|i-j| \ge d' $.
This implies that there cannot be any pair of indices $(i', j')$ in $M$ such that $i' = j'$, for otherwise there are already at least $d'$ bits in $C_1$ or $C_2$ that are not matched.
Let $\hat{t}$ be the larger number of non-zero blocks in $w$ and $w'$. Note that $\hat{t} \geq s$.
Since all inner codes are independent and random, and every pair of indices $(i, j)$ in $M$ has $i\neq j$, the probability that $M$ is a matching is at most $2^{-((\hat{t}-1) n' - O(d'))}$ ($w'$ can have length as small as $(t+1)n-n-d'$). Now if we apply the union bound, the main term is actually the total number of possible tuples of the non-zero inner codewords. Since there are at most $2\hat{t}$ non zero blocks in $w$ and $u$, this number is at most $2^{2\hat{t}k'}$. Thus as long as $s$ is a large enough integer, the rate of the code can approach $1/2$.  

The second case is that every pair of indices $(i, j )$ in $M$ has $|i-j| < d'$.
In this case we focus on the unique blocks. Let $s'$ be the larger number of unique blocks in $w$ and $u$, and for simplicity assume $u$ has more unique blocks. We delete all matches where the endpoint in $u$ is not in a unique block, or the endpoint in $w$ falls out of the block at the same position as the block in $u$ which contains the endpoint in $u$. Thus we attain a trimmed alignment $M'$.
Under the assumed condition in this case, we don't lose too many matches. Indeed the number of matches left is at least $ \ell' = s'(n'-d') -n'-d'$. We now upper bound the probability that there exists such an $M'$ which is a valid matching. Since this event is implied by the original event, this also provides an upper bound of the original event. 

The probability that any $M'$ is a valid matching is $2^{\ell'}$, by our definition of unique blocks. Now in the union bound, the main term turns out to be the total number of possible tuples of the inner codewords corresponding to the $s'$ unique blocks in $u$ and the other $s'$ blocks at the same positions in $w$, which is roughly $2^{(2s')k'}$. Notice that $s' \geq s$. Thus in this case, as long as $s$ is large enough, the rate of the code can also approach $1/2$.  

\subparagraph{The existence of linear concatenated codes matching random linear codes.} In section~\ref{section:existence} we show the existence of linear concatenated insdel codes that match the parameters of random linear codes. This is similar in spirit to Thommesen's work \cite{1056765}, which shows the existence of binary linear concatenated codes with Reed-Solomon outer codes that asymptotically meet the Gilbert-Varshamov bound. In particular, we also take a Reed-Solomon code as the outer code, and use an independent random linear inner code for every symbol of the outer codeword. Interestingly, here we take the outer code to be a $[n, k=(1-\gamma)n/2, d=(1+\gamma) n/2]_q$ Reed-Solomon code with $q=\Theta(n)$, i.e., the rate of the outer code is less than $1/2$. On the other hand, we take all inner codes to have rate $1$. Using a careful probabilistic counting argument together with an estimate of the number of Reed-Solomon codewords with a specific weight (as done in \cite{1056765}), we can prove the existence of linear concatenated insdel codes with parameters as in Theorem~\ref{thm:existence}. 

The choice of the parameters of the outer code is different from our explicit constructions, suggesting that maybe different constructions based on these parameters can lead to better explicit linear insdel codes.

%% file: organization.tex
\paragraph{Organization of the rest of the paper.} After giving some notation and previous works used in Section~\ref{sec:prelim}, we present our main results in the subsequent sections. Specifically, in Section~\ref{sec:higherror} we give our codes for high noise. In Section~\ref{sec:high rate 1} and Section~\ref{sec:high rate 2} we give our codes with high rate. In Section~\ref{section:existence} we prove the existence of linear concatenated insdel codes that match the parameters of random linear codes. In section~\ref{sec:lowerbound} we prove the weaker version of our conjecture. Finally we conclude with some open problems in section~\ref{sec:open}.

%% file: prelim.tex
\section{Preliminaries}\label{sec:prelim}
\textbf{Notation.} Let $\Sigma$ be an alphabet. For a string $x\in \Sigma^*$,
\begin{enumerate}
\item $|x|$ denotes the length of the string.
\item $x[i,j]$ denotes the substring of $x$ from position $i$ to position $j$ (both endpoints included).
\item $x[i]$ denotes the $i$-th symbol of $x$.

\item $x\circ x'$ denotes the concatenation of $x$ and some other string $x'\in \Sigma^*$.
\item For a string $s$ which is a concatenation of  shorter strings $s_1, s_2, \ldots, s_t$,   the $i$-th block of $s$ refers to $s_i$.
\end{enumerate}
%
%
\begin{definition}[Edit distance and Longest Common Subsequence] For any two strings $x, y\in\Sigma^n$, the edit distance $\ed(x, y)$ is the minimum number of edit operations (insertions and deletions) required to transform $x$ into $y$.\footnote{The standard definition of edit distance also allows substitution, but for simplicity we only consider insertions and deletions here, as a substitution can be replaced by a deletion followed by an insertion.} A longest common subsequence of $x$ and $y$ is a longest pair of subsequences of $x$ and $y$ that are equal as strings. We use $\lcs(x, y)$ to denote the length of a longest common subsequence between $x$ and $y$.
\end{definition}
Note that $\ed(x, y) = |x|+|y|-2\cdot \lcs(x, y)$. We use $\hamming(x, y)$ to denote the Hamming distance between two strings $x$ and $y$.

\begin{definition}
An $(n, m ,d)$-code $C$ is an error-correcting code (for Hamming errors) with codeword length $n$, message length $m$, such that the Hamming distance between every pair of codewords in $C$ is at least $d$.
\end{definition}
%
%
\begin{definition}
Fix an alphabet $\Sigma$, an error-correcting code $C\subseteq \Sigma^n$ for edit errors with message length $m$ and codeword length $n$ consists of an encoding function $\text{Enc}:\Sigma^m\rightarrow \Sigma^n$ and a decoding function $\text{Dec}:\Sigma^*\rightarrow \Sigma^m$. The code can correct $k$ edit errors if for every $y$, s. t. $\ed(y,\text{Enc}(x))\leq k$, we have $\text{Dec}(y) = x$. The rate of the code is defined as $\frac{m}{n}$. 
\end{definition}

We say $C$ is a linear code if the alphabet $\Sigma$ is a finite field $\F_q$ and the encoding function $\text{Enc}:\F_q^m\rightarrow \F_q^n$ is a $\F_q$-linear map.

We use the following list decoding algorithm for Reed-Solomon codes due to Guruswami and Sudan \cite{GuruswamiS99}.

\begin{theorem}\label{thm:listdecoding}
Given a family of Reed-Solomon codes of message rate $\gamma$, an error rate of $\eps=1-\sqrt{\gamma}$ can be list-decoded in polynomial time.
\end{theorem}

We use $U_n$ to denote the uniform distribution on $\{0,1\}^n$.

\begin{definition}
An $\eps$-biased distribution $X$ over $\{0, 1\}^n$ is such that for any $S \subseteq [n]$, $|\Pr\left[ \bigoplus_{i\in S} X_i = 1\right] - 1/2| \le \eps$. 
A function $g: \{0, 1\}^s \rightarrow \{0, 1\}^n$ is an $\eps$-biased generator if $g(U_s)$ is an $\eps$-biased distribution.
\end{definition}
The following $\eps$-biased generator is used.
\begin{theorem}[\cite{alon1992simple}]\label{thm:epsbiased}
For every $n\in\mathbb{N}$, every $\eps \in (0,1)$, there exists an explicit $\eps$-biased generator $\{0, 1\}^s \xrightarrow{} \{0, 1\}^n $ with $s = O( \log n +   \log( 1/\eps) )$. 
\end{theorem}
We also need the following XOR lemma.
\begin{lemma}[XOR Lemma]
\label{lem:XORlem}
The statistical distance between an $\eps$-biased distribution and a uniform distribution, both over $\{0, 1\}^n$,  is at most $ \eps  \sqrt{2^n}$.
\end{lemma}

%% file: general_construction.tex
\section{General Construction of Our Codes}
All our codes follow the general strategy of code concatenation, which we describe below.

The outer code $\Cout$ with encoding function $:\enc_{out}: \alphabetout^k \rightarrow \alphabetout^n$  is an  $[n,k,d]$ Reed Solomon Code for Hamming errors.
We then use $n$ different inner codes $\Cin^1, \dots, \Cin^n$, such that for any $i \in [n]$, $\Cin^i$ is a linear code $\enc_{in}^i:\alphabetout \rightarrow \alphabetin^{n'}$, where $n'$ is the block length of the inner code. In this paper $\alphabetin$ always has constant size and we let $n' = \Theta(\log n)$. For different applications, we will need the inner codes to have slightly different properties.

Our final code $\C$ works naturally by first encoding the message using the outer code, then encoding each symbol of the outer code using the inner codes. This gives a codeword over $\alphabetin$ with length $N=n\cdot n'$. If the outer code and all the inner codes are linear, the concatenated code is also linear.

%% file: construction_herror.tex
\section{Constructions For High Noise}\label{sec:higherror}
In this section we give our linear codes that can correct $1-\eps$ fraction of insdel errors, for any constant $\eps>0$. Our codes can still achieve a constant rate. 

\paragraph*{The construction.} Following our general construction, here we take $\Cout$ to be an $[n,k,d]_n$ Reed-Solomon code with alphabet size $|\alphabetout| = n$, $k =\gamma n$ and $d=(1-\gamma) n$ for some constant $\gamma>0$ to be chosen later. We construct $n$ different inner codes $\Cin^1, \dots, \Cin^n$ with alphabet size $|\alphabetin| = \poly(1/\gamma)$, message length $k'=\Theta(\log n)$, and codeword length $n'=\Theta(\log n)$, with the following property.

\begin{property}\label{prop:property1}
For any two codewords $x \in \Cin^i, y \in \Cin^j$, if either of the following two conditions holds:
\begin{enumerate}
    \item $i \neq j$, and $x \neq 0^{n'}$ or $y \neq 0^{n'}$.
    \item $i = j$ and $x \neq y$.
\end{enumerate}
Then we have $\lcs(x, y) \leq \gamma n'$.
\end{property}

\begin{lemma}\label{lem:property1}
There exists an efficient construction of $n$ inner codes $\Cin^1, \dots, \Cin^n$, where each $\Cin^i$ has alphabet size $|\alphabetin| = \poly(1/\gamma)$ and rate $\Omega(\gamma)$.
\end{lemma}

\begin{proof}
We first show that if we pick $n$ independent random linear inner codes $\Cin^1, \dots, \Cin^n$ over an alphabet size $|\alphabetin| = \poly(1/\gamma)$, then they satisfy \cref{prop:property1} with high probability. We then show how to derandomize the construction using a small biased sample space.

Fix a field $\F_q$. For each $\Cin^i$ we independently pick $\log n$ uniformly random vectors in $\F_q^{n'}$ with $n'=\Theta(\log n/\gamma)$ as the basis for $\Cin^i$, or equivalently, the rows in the generating matrix of $\Cin^i$. We bound the probability that there exist two codewords $x \in \Cin^i, y \in \Cin^j$ that satisfy the conditions of \cref{lem:property1} but $\lcs(x, y) > \gamma n'$.

We have the following claim.

\begin{claim}
\label{claim:pairwise-ind}
Consider any fixed common subsequence between $x$ and $y$ of length $t$, where the corresponding indices in $x$ are $\{s_1, \cdots, s_t\}$ and the corresponding indices in $y$ are $\{r_1, \cdots, r_t\}$. Then 
\[\Pr[\forall k \in [t], x_{s_k}=y_{r_k}] \leq q^{-t}.\]
\end{claim}

To prove the claim we have two cases.
\begin{description}
\item [Case 1:] $i \neq j$, and $x \neq 0^{n'}$ or $y \neq 0^{n'}$. This is the easy case. Since $i \neq j$, and all the entries in the generating matrices of $\Cin^i$ and $\Cin^j$ are chosen independently uniformly from $\F_q$, we know that the events $x_{s_k}=y_{r_k}$ are all independent, even if $x = 0^{n'}$ or $y = 0^{n'}$. Furthermore, the probability of each such event is $1/q$. Hence the claim follows.
\item [Case 2:] $i = j$. In this case, the events $x_{s_k}=y_{r_k}$ are not necessarily all independent. However, the claim still follows from the following claim in \cite{CGHL21}, which deals exactly with this situation.

\begin{claim}\label{clm:samecode}[Claim 4.2 of \cite{CGHL21}]
Let $G$ be a random generating matrix for a linear code over $\F_q$. For any two different messages $x^i, x^j$ and codewords $C^i=x^i G, C^j=x^j G$, consider any fixed common subsequence between $C^i$ and $C^j$ of length $t$, where the corresponding indices in $C^i$ are $\{s_1, \cdots, s_t\}$ and the corresponding indices in $C^j$ are $\{r_1, \cdots, r_t\}$. Then 
\[\Pr[\forall k \in [t], C^i_{s_k}=C^j_{r_k}] \leq q^{-t}.\]
\end{claim}
\end{description}

Now by a standard union bound, and noticing that the total number of possible cases where two strings of length $n'$ have a common subsequence of length $\gamma n'$ is at most $\binom{n'}{\gamma n'}^2$, we have
\[\Pr[\text{\cref{prop:property1} does not hold}] \leq  n^2 q^{2 \log n} \binom{n'}{\gamma n'}^2 q^{-\gamma n'}\leq  \left (\frac{e}{\gamma} \right )^{2\gamma n'} n^2 q^{2 \log n-\gamma n'}.\]
Therefore, one can set $q=(\frac{e}{\gamma})^3$ and $n'=\Theta(\log n/\gamma)$ so that the above probability is $q^{-\Omega(\log n)}=1/\poly(n)$.

Next we show how to derandomize the above construction using a small biased space. Without loss of generality we assume the field we use is $\F_q$ with $q=2^{\ell}$. Thus, by choosing an arbitrary basis $b_1, \cdots, b_{\ell}$ in $\F_q$ we can identify the field with the vector space $\F_2^{\ell}$, such that any $a \in \F_q$ can be expressed as $a =\sum_{i \in [\ell]} a_i b_i$, where $\forall i, a_i \in \F_2$. In this way, the generating matrix of each $\Cin^i$ can be viewed as consisting of $\ell n' \log n=\Theta(\ell \log^2 n )$ bits.

We pick a $\tau$-biased sample space with $n \ell n' \log n $ bits for some $\tau=1/\poly(n)$ to be chosen later. Note that by Theorem~\ref{thm:epsbiased} this can be generated by $O(\log n)$ uniform random bits.

Given $\ell$ bits $a_1, \cdots, a_{\ell}$ which defines the field element $a =\sum_{i \in [\ell]} a_i b_i$, and any $p \in \F_q$, consider the operation $p \cdot a$ and the corresponding coefficient in the basis $b_1$. It's not hard to see that this is a $\F_2$-linear function (i.e., a parity) of $a_1, \cdots, a_{\ell}$. Call this parity $L_p (a_1, \cdots, a_{\ell})$. We have the following claim.

\begin{claim}\label{clm:parity}
$L_p (a_1, \cdots, a_{\ell}) \equiv 0$ if and only if $p=0$.
\end{claim}

\begin{proof}[Proof of the claim.]
The ``if" part is trivially true. For the other part, note that if $p \neq 0$ then $p b_1, \cdots, p b_{\ell}$ must also be linearly independent and thus form a basis of $\F_q$. Therefore, some $p b_i$ must have a non-zero coefficient in $b_1$ and thus $L_p (a_1, \cdots, a_{\ell})$ has a term $a_i$ in the parity, therefore it cannot be the $0$ function.
\end{proof}

Note that there are altogether $2^{\ell}$ different parity functions involving $a_1, \cdots, a_{\ell}$, and $q=2^{\ell}$ elements in $\F_q$. Thus the previous claim also immediately implies the following claim.

\begin{claim}\label{clm:parity2}
Any parity function involving $a_1, \cdots, a_{\ell}$ is equivalent to $L_p (a_1, \cdots, a_{\ell})$ for some $p \in \F_q$.
\end{claim}

Now consider the two codewords $x \in \Cin^i, y \in \Cin^j$. Let $x_0$ and $y_0$ be the corresponding messages for $x$ and $y$ respectively. We now have the following claim.

\begin{claim}
Unless $i=j$ and $y_0 = p \cdot x_0$ or $x_0 = p \cdot y_0$ for some $p \in \F_q$, under the $\tau$-biased sample space, the joint distribution of $(x, y)$ is $q^{n'}\tau$-close to the uniform distribution over $\F_q^{2n'}$.
\end{claim}

\begin{proof}[Proof of the claim.]
Let $x=(x_1, \cdots, x_{n'}) \in \F_q^{n'}=\F_2^{\ell n'}$ and $y=(y_1, \cdots, y_{n'}) \in \F_q^{n'}=\F_2^{\ell n'}$. Consider any non-trivial parity of the $2 \ell n'$ bits, which by \cref{clm:parity2} corresponds to the coefficient of $b_1$ under some function $\sum_{k \in [n']} (p^x_k x_k+p^y_k y_k)$, where $\forall k, p^x_k, p^y_k \in \F_q$, and they are not all $0$.

If $i \neq j$, then $\sum_{k \in [n']} p^x_k x_k$ and $\sum_{k \in [n']} p^y_k y_k$ use different bits in the $\tau$-biased sample space. Since $x, y$ are not both $0^{n'}$, the resulted parity is a non-trivial parity of the bits in the sample space, which by definition has bias at most $\tau$.

Otherwise we have $i=j$. Let $G$ be the generating matrix for $\Cin^i$, thus $x=x_0 G$ and $y=y_0 G$. For any $k \in [n']$, let $G_k$ be the $k$'th column of $G$. We have

\[\sum_{k \in [n']} (p^x_k x_k+p^y_k y_k)=\sum_{k \in [n']} (p^x_k x_0 G_k+p^y_k y_0 G_k)= \sum_{k \in [n']} (p^x_k x_0+p^y_k y_0) G_k.\]
Notice that each entry in each $G_k$ is independently uniformly chosen from $\F_q=\F_2^{\ell}$. Thus by \cref{clm:parity} if the coefficient of the above sum in $b_1$ is the trivial parity $0$, then we must have $\forall k \in [n'], p^x_k x_0+p^y_k y_0=0$. This implies that either $y_0 = p \cdot x_0$ or $x_0 = p \cdot y_0$ for some $p \in \F_q$. 

Otherwise, the parity is a non-trivial parity of the bits in the sample space, which by definition has bias at most $\tau$. Now, by Lemma~\ref{lem:XORlem}, the joint distribution of $(x, y)$ is $q^{n'}\tau$-close to the uniform distribution over $\F_q^{2n'}$.
\end{proof}

Back to the proof of our lemma. If the conditions of the above claim hold, then the joint distribution of $(x, y)$ is $q^{n'}\tau$-close to the uniform distribution. Hence, the probability that there exists any common subsequence of length $\gamma n'$ between $x$ and $y$ is at most $\binom{n'}{\gamma n'}^2 q^{-\gamma n'}+q^{n'}\tau$.

On the other hand, if the conditions of the above claim do not hold, then without loss of generality assume that $y_0 = p \cdot x_0$ for some $p \in \F_q$. Hence $p \neq 1$. In this case, notice that we also have $y = p \cdot x$, and thus the probability that there exists any common subsequence of length $\gamma n'$ between $x$ and $y$ is completely determined by the random variables in $x$. Note that any non-trivial parity of the bits in $x$ is also a non-trivial parity of the bits of the $\tau$-biased sample space, which has bias at most $\tau$. By Lemma~\ref{lem:XORlem}, the distribution of $x$ is $q^{n'/2}\tau$-close to being uniform on $\F_q^{n'}$.

We have the following claim.

\begin{claim}
Let $x$ be a uniformly random vector in $\F_q^{n'}$, and $y = p \cdot x$. Then 
\[\Pr[\exists \text{a common subsequence of length } t \text{ between } x \text{ and } y] \leq  \binom{n'}{t}^2 q^{-t}.\]
\end{claim}

\begin{proof}[Proof of the claim.]
Consider any fixed common subsequence of length $t$ between $x$ and $y$. Assume where the corresponding indices in $x$ are $\{s_1, \cdots, s_t\}$ and the corresponding indices in $y$ are $\{r_1, \cdots, r_t\}$, such that $s_1 < s_2 < \cdots < s_t$ and $r_1 < r_2 < \cdots < r_t$. For any $k \in [t]$, let $m_k= \max(s_k, r_k)$. Notice that $m_1 < m_2 < \cdots < m_t$. Define $E_k$ to be the event $x_{s_k}=y_{r_k}$.

For each $k \in [t]$, if $s_k=r_k$, then 
\[\Pr[E_k] = \Pr[x_{s_k}=p \cdot x_{s_k}] =\Pr[x_{s_k}=0]=\frac{1}{q}. \]

Furthermore, since $s_k=r_k=m_k$ is larger than all $\{s_{k'}, r_{k'}, k' < k\}$, the event $E_k$ is independent of all $\{E_{k'}, k' < k\}$. Thus 
\[\Pr[E_k | \{E_{k'}, k' < k\}] =\frac{1}{q}.\]

Otherwise, $s_k \neq r_k$ and without loss of generality assume $s_k > r_k$. This means $s_k=m_k$ and is larger than all $\{s_{k'}, r_{k'}, k' < k\}$. We can now first fix all $\{x_{s_{k'}}, y_{r_{k'}}, k' < k\}$ and $y_{r_k}$, and conditioned on this fixing $x_{s_k}$ is still uniform over $\F_q$. Thus

\[\Pr[E_k] = \Pr[x_{s_k}=p \cdot x_{r_k}] =\frac{1}{q}. \]

Note that any such fixing also fixes the outcomes of all $\{E_{k'}, k' < k\}$. Hence we also have

\[\Pr[E_k | \{E_{k'}, k' < k\}] =\frac{1}{q}.\]
Therefore, the above equation holds in all cases, and for all $k$. This gives 
\[\Pr[\bigcap_{k \in [t]} E_k] \leq q^{-t},\]
and the claim follows from a union bound.
\end{proof}
Since $x$ is $q^{n'/2}\tau$-close to being uniform on $\F_q^{n'}$, we have that the probability that there exists any common subsequence of length $\gamma n'$ between $x$ and $y$ is at most $\binom{n'}{\gamma n'}^2 q^{-\gamma n'}+q^{n'/2}\tau$ in this case.

To summarize, using the $\tau$-biased sample space we always have that the probability that there exists any common subsequence of length $\gamma n'$ between $x$ and $y$ is at most $\binom{n'}{\gamma n'}^2 q^{-\gamma n'}+q^{n'}\tau$. By another union bound, we have

\begin{align*}
\Pr[\text{\cref{prop:property1} does not hold}] & \leq  n^2 q^{2 \log n} \left (\binom{n'}{\gamma n'}^2 q^{-\gamma n'}+q^{n'}\tau \right) \\ &\leq  \left (\frac{e}{\gamma} \right )^{2\gamma n'} n^2 q^{2 \log n-\gamma n'}+n^2 q^{n'+2 \log n} \tau.    
\end{align*}

Therefore, one can still set $q=(\frac{e}{\gamma})^3$, $n'=\Theta(\log n/\gamma)$, and $\tau=q^{-\Omega(\log n/\gamma)}=1/\poly(n)$ so that the above probability is $q^{-\Omega(\log n)}=1/\poly(n)$.

Once we know this, we can exhaustively search the $\tau$-biased sample space and find a sample point which gives us a construction that satisfies \cref{prop:property1}. Since we only have $\poly(n)$ sample points and checking each sample point takes polynomial time, altogether this takes polynomial time.
\end{proof}

Note that our concatenated code $\C$ now has rate $\Omega(\gamma^2)$. Further, \cref{prop:property1} implies the following property:
\begin{property}\label{prop:property1a}
\begin{enumerate}
    \item $\forall i \neq j$, we have $\Cin^i \bigcap \Cin^j = \{0^{n'}\}$.
    \item For any $i, j \in [n]$ and any two codewords $x \in \Cin^i, y \in \Cin^j$, if $x \neq y$ then $\lcs(x, y) \leq \gamma n'$.
\end{enumerate}
\end{property}

\paragraph*{Notation.} Let $z$ be any substring of a codeword from the concatenated code $\C$, and assume $z$ is a substring of $z_{j} \circ z_{j+1} \circ \cdots \circ z_{j+\ell}$, where $\forall t$, $z_{j+t}$ is a codeword in $\Cin^{j+t}$. We say the codewords $\{z_{j+t}, t=0, \cdots, \ell\}$ contribute to the string $z$.

We now show that \cref{prop:property1} and Property~\ref{prop:property1a} give us the following lemma.

\begin{lemma}\label{lem:property2}

Let $x$ be a codeword from the code $\Cin^i$. Let $z$ be any substring of a codeword from the concatenated code $\C$, and $\{z_{j+t}, t=0, \cdots, \ell\}$ are the inner codewords contributing to $z$. If $\forall t$, $z_{j+t} \neq x$, 
then we have $\lcs(x, z) < 2\gamma (|x|+|z|)$.
\end{lemma}

\begin{proof}
By Property~\ref{prop:property1a}, the longest common subsequence between $x$ and any $z_{j+t}$ has length at most $\gamma n'$. If $\ell=1$, then we have 
\[\lcs(x, z) \leq \gamma n' < 2 \gamma (|x|+|z|).\]

Otherwise we have $\ell \geq 2$. Notice that $|z| > (\ell-2)n'$. Thus we have 
\[\lcs(x, z) \leq \ell \gamma n' \leq 2 \gamma (\ell-1)n' < 2 \gamma(|x|+|z|).\]
\end{proof}

We can now prove the following lemma.

\begin{lemma}\label{lem:distance}
For any two different codewords $C_1, C_2 \in \C$, we have $\ed(C_1, C_2) > 2(1-6 \gamma)N$.
\end{lemma}

\begin{proof}
We upper bound $\lcs(C_1, C_2)$ as follows. Consider a particular longest common subsequence and divide it sequentially according to the $n$ inner codewords in $C_1$. Let the codewords in $C_1$ be $x_1, \cdots, x_n$ and the corresponding substrings in $C_2$ under the LCS be $z_1, \cdots, z_n$.

By Lemma~\ref{lem:property2}, for any $i \in [n]$, we must have $\lcs(x_i, z_i) \leq 2 \gamma(|x_i|+|z_i|)$, unless some inner codeword in $z_i$ is equal to $x_i$. This could happen either because $x_i=0^{n'}$ or because $z_i$ contains part of $x_i$ from exactly the $i$'th inner code. In the latter two cases, we call such an index $i$ \emph{bad}. Note that for a bad $i$ we have $\lcs(x_i, z_i) \leq n'$, and there are at most $\gamma n$ such bad indices for either case, by our choice of the outer code. Let $t$ be the number of bad indices, thus $t \leq 2 \gamma n$. Therefore, 

\begin{align*}
\lcs(x, z) &= \sum_{i \text{ is not bad}} \lcs(x_i, z_i)+\sum_{i \text{ is bad}} \lcs(x_i, z_i) \\ &\leq 2 \gamma \sum_{i \text{ is not bad}}(|x_i|+|z_i|)+t n'  \\ & < 2 \gamma (2 n'n)+  2 \gamma nn' = 6 \gamma N,
\end{align*}
where the last inequality follows from the fact that if the number of bad indices is larger than $0$, then $\sum_{i \text{ is not bad}}(|x_i|+|z_i|) < 2n'n$. 
Therefore $\ed(C_1, C_2) > 2N-12\gamma N=2(1-6\gamma)N$.
\end{proof}

Setting $\gamma=\eps/6$, this gives the following theorem.

\begin{theorem}
For any constant $\eps>0$ there exists an efficient construction of linear insdel codes over an alphabet of size $\poly(1/\eps)$, with rate $\Omega(\eps^2)$ that can correct $1-\eps$ fraction of insdel errors (possibly inefficiently).
\end{theorem}


\subsection{Efficient Decoding}
We now give an efficient decoding algorithm for the above code. The algorithm is simple. Given any received string $y$, we consider the partition of $y$ into $n$ substrings $y_1, \cdots, y_n$ such that $y =y_1 \circ y_2 \circ \cdots \circ y_n$, where each $y_i$ can be the empty string. Now for each $y_i$, we find the closest codeword $x_i \in \Cin^i$ in insdel distance and record $\Delta_i=\ed(y_i, x_i)$. Finally we compute $\Delta=\sum_i \Delta_i$. Our goal is to minimize $\Delta$, and this can be done by using the following dynamic programming. 

We define a $n \times (N+1)$ matrix $f$ such that $\forall i \in [n], j \in [N] \cup \{0\}$, $f[i,j]$ records the optimal sum of insdel distances between a partition of $y_{[1:j]}$ into $i$ substrings and the concatenation of the first $i$ inner codewords. 

For any $0 \le p,  q \le N$, we define the function $g_i(y_{[p, q]})$ to be the insdel distance of the closest codeword $x_i \in \Cin^i$ to $y_{[p, q]}$. Note that the function $g_i$ can be computed in polynomial time by enumerating the $n$ codewords in $\Cin^i$ and computing the corresponding insdel distance to $y_{[p, q]}$. Now for $1\le i \le n$ and $0\le j \le N$, we can compute $f[i,j]$ as follows.

First, $\forall j \in  [N] \cup \{0\}$, let $f[1, j]=g_1(y_{[1, j]})$. Then, for $2\le i \le n$ and $ 0 \le j \le N$, we compute $f[i,j]$ using the following updating rule:

\[f[i,j] = 
\underset{0\le j' \le j}{\min} \; f[i-1, j'] + g_i(y_{[j'+1: j]}).
\]

Finally we return $f[n, N]$. To recover the inner codewords $\{x_i, i \in [n]\}$, during the dynamic programming we also maintain another $n \times (N+1)$ matrix $M$, such that $\forall i \in [n], j \in [N] \cup \{0\}$, $M[i, j]$ records two items: the right index of the previous substring (i.e., the optimal $j'$ in the updating rule), and the closest codeword $x_i \in \Cin^i$ to the substring $y_{[j'+1: j]}$. With these information one can trace back from $M[n, N]$ and recover all the inner codewords. 

After this, we use each $x_i$ to recover the corresponding symbol of the outer code at the $i$'th coordinate, and then run the list decoding algorithm from \cref{thm:listdecoding}. From the returned list, we exhaustively search for the codeword which has insdel distance less than $(1-6\gamma)N$ from the received string $y$.

We have the following lemma.

\begin{lemma}\label{lem:errorbound}
Assume the received string $y$ has at most $\alpha N$ insdel errors, then the above algorithm recovers at least $((1-2 \gamma)-\frac{\alpha}{1-2 \gamma}) n$ correct symbols in the outer codeword.
\end{lemma}

\begin{proof}
First note that there is a natural partition of $y$ according to the actual inner codewords. With this partition and let each $x_i$ be the corresponding correct inner codewords, one can see that the total insdel distance is at most $\alpha N$. Therefore, the optimal partition obtained by the dynamic programming must have $\Delta=\sum_i \Delta_i \leq \alpha N$.

Assume there are $\beta n$ symbols of the outer codeword that are recovered correctly after the algorithm. Let $z$ be the actual codeword transmitted without errors, and let $z_1, \cdots, z_n$ be the corresponding strings of $y_1, \cdots, y_n$ before errors, such that $z=z_1 \circ z_2 \circ \cdots \circ z_n$. 

Note that at most $2\gamma n$ of the $z_i$'s contains part of the codeword $0^{n'}$. Hence there are at least $(1-\beta-2 \gamma)n$ of the $z_i$'s that don't have the codeword $0^{n'}$ contributing to them, and are recovered to an incorrect inner codeword $x_i$ in the algorithm. For simplicity, let's pick exactly $(1-\beta-2 \gamma)n$ of such $z_i$'s.

We claim that any such $z_i$ and the incorrectly recovered inner codeword $x_i$ satisfy the conditions of Lemma~\ref{lem:property2}. To see this, notice that first of all, the codeword $0^{n'}$ does not contribute to $z_i$. Hence if $x_i=0$, then $x_i$ and $z_i$ satisfy the conditions of Lemma~\ref{lem:property2}. Otherwise, $x_i \neq 0$. If no codeword from $\Cin^i$ contributes to $z_i$, then by Property~\ref{prop:property1a}, $x_i$ and $z_i$ satisfy the conditions of Lemma~\ref{lem:property2}. Else, some codeword $x'_i \in \Cin^i$ contributes to $z_i$ and thus $x'_i$ is the correct inner codeword from $\Cin^i$. Since $x_i$ is recovered incorrectly, we must have $x_i \neq x'_i$, and again by Property~\ref{prop:property1a}, $x_i$ and $z_i$ satisfy the conditions of Lemma~\ref{lem:property2}. 

Therefore, let $B \subseteq [n]$ denote the index set of the $(1-\beta-2 \gamma)n$ of the $z_i$'s, and let $L=\sum_{i \in B} |z_i|$ be the total length of these $z_i$'s, by Lemma~\ref{lem:property2} we know that $\forall i \in B$, $\lcs(x_i, z_i) \leq 2\gamma(|x_i|+|z_i|)$. Thus

\[\sum_{i \in B} \ed(x_i, z_i) \geq \sum_{i \in B} (1-2\gamma)(|x_i|+|z_i|) \geq (1-2\gamma)((1-\beta-2 \gamma)N+L).\]

We now have two cases.

\begin{description}
\item [Case 1:] $L \geq (1-\beta-2 \gamma)N$. Then we have 
\[\sum_{i \in [n]} \ed(x_i, z_i) \geq \sum_{i \in B} \ed(x_i, z_i) \geq 2 (1-2\gamma)(1-\beta-2 \gamma)N.\]
\item [Case 2:] $L < (1-\beta-2 \gamma)N$. Then we have $\sum_{i \in [n] \setminus B} |z_i| = N-L$ and thus 
\[\sum_{i \in [n] \setminus B} \ed(x_i, z_i) \geq  \sum_{i \in [n] \setminus B} |z_i|- \sum_{i \in [n] \setminus B} |x_i| = N-L-(\beta+2 \gamma)N=(1-\beta-2 \gamma)N-L.\]

Thus
\begin{align*}
    \sum_{i \in [n]} \ed(x_i, z_i) &= \sum_{i \in B} \ed(x_i, z_i)+\sum_{i \in [n] \setminus B} \ed(x_i, z_i) \\ &\geq (1-2\gamma)((1-\beta-2 \gamma)N+L)+(1-\beta-2 \gamma)N-L \\ & > 2 (1-2\gamma)(1-\beta-2 \gamma)N.
\end{align*}
\end{description}
To conclude, we always have $\sum_{i \in [n]} \ed(x_i, z_i) \geq 2 (1-2\gamma)(1-\beta-2 \gamma)N.$ Now, since there are at most $\alpha N$ insdel errors, and the algorithm returns a matching such that $\Delta=\sum_i \Delta_i=\sum_i \ed(y_i, x_i) \leq \alpha N$, by triangle inequality we must have

\[2 (1-2\gamma)(1-\beta-2 \gamma)N \leq \sum_{i \in [n]} \ed(x_i, z_i) \leq \sum_{i \in [n]} \ed(x_i, y_i)+\alpha N \leq 2 \alpha N.\]

This gives 

\[\beta \geq (1-2 \gamma)-\frac{\alpha}{1-2 \gamma}.\]
\end{proof}

We can now prove the following theorem.

\begin{theorem}
For any constant $\gamma >0$, the code $\C$ can correctly decode from up to $1-6\sqrt{\gamma}$ fraction of insdel errors in polynomial time.
\end{theorem}

\begin{proof}
By \cref{lem:errorbound} and \cref{thm:listdecoding}, the list decoding algorithm of the outer code succeeds as long as 

\[(1-2 \gamma)-\frac{\alpha}{1-2 \gamma} \geq \sqrt{\gamma}.\]

At the same time, by \cref{lem:distance} the unique correct codeword can be found by examining the list of returned codewords as long as $\alpha \leq 1-6\gamma$.

One can check that when $\alpha \leq 1-6\sqrt{\gamma}$, both conditions are satisfied.
\end{proof}

By setting $\gamma=\frac{\eps^2}{36}$, this immediately gives the following corollary.

\begin{corollary}
For any constant $\eps >0$, there is a linear code $\C$ with rate $\Omega(\eps^4)$ and alphabet size $\poly(1/\eps)$, that can be encoded in polynomial time and decoded from up to $1-\eps$ fraction of insdel errors in polynomial time.
\end{corollary}

%% file: construction_hrate1.tex
\section{Constructions with High Rate}\label{sec:high rate 1}
In this section, we  give our linear codes with a rate arbitrarily close to $1/3$, which can still correct a constant fraction of insdel errors.

\begin{theorem}\label{thm:highrate1}
For every constant $\gamma \in (0, 1/3)$, there exists a binary linear insdel code with rate $\ge 1/3 - \gamma$, relative decoding radius $  \Omega( \frac{\gamma^2}{\log \frac{1}{\gamma}})$, where both the encoding and decoding are in polynomial time.
\end{theorem}


\subsection{Construction} 

\paragraph{The construction} Following our general construction, given a constant $\gamma \in (0, 1/3)$, setting $\delta = \gamma/2$, we take $\Cout$ to be an $[n,k,d]_n$ Reed-Solomon code with $|\alphabetout| = n$, $k = (1- \delta) n$ and $d = \delta  n$. 
We can then construct $n$ different binary inner codes $\Cin^1, \dots, \Cin^n$ with message length $k'$ and codeword length $n'$ such that $ k'= \log n$ and $ n' =  O(\frac{k'}{\gamma})$. We need these inner codes satisfying the following property, i.e. \cref{prop:propertymainbasic}. 
Finally, for each $i\in [n]$ we replace the alphabet of $\Cout$ in the $i$-th block with $\Cin^{i}$.

\begin{property}[with parameter $d'$]
\label{prop:propertymainbasic}
$\forall i \in [n], j \in [n-1]$, $\forall w\in \Cin^i \backslash \{0^{n'}\}, u\in \Cin^j, v\in \Cin^{j+1}$, unless $w = u$ or $w = v$, it always holds that for any substring $ w' $ of $u\circ v$, $\ed(w, w')\ge d'$.

We say the property is with parameter $d'$ to indicate the parameter $d'$.
\end{property}

The following lemma shows an explicit construction of inner codes satisfying \cref{prop:propertymainbasic} with a parameter $d' = \Omega( \frac{\gamma}{\log \frac{1}{\gamma}}n')$.


\begin{lemma}\label{lem:innerconstruct}
For every $\gamma \in (0, 1/3)$,  there exists an efficient construction of $n$ binary inner codes $\Cin^1, \dots, \Cin^n$ satisfying \cref{prop:propertymainbasic} with parameter $d' = \frac{\gamma}{ 10 \log \frac{1}{\gamma}}  \cdot n' , n' = \frac{10\log n}{\gamma}$, where each $\Cin^i$ has rate $  \frac{1}{3}-\gamma $.
\end{lemma}


\begin{proof}
We first show that a random construction can give the desired codes with a good probability. 
Then show how to derandomize.

For every $i$, the random construction independently chooses $(\frac{1}{3}-\gamma)n'$ vectors in $\{0,1\}^{n'}$ to be the basis for $\Cin^i$.

Consider $w, u, v, w'$ in \cref{prop:propertymainbasic}.
If  $ \ed(w, w') \le d' $, then $\lcs(w, w') \geq  (n' + |w'|-d')/2$.
Let $M$ be a monotone alignment of size $m  = \lcs(w, w')$.
Now consider the probability that the alignment is indeed a matching.
Since $w\neq u, w\neq v$, as  every base vector is independently randomly chosen, $w, u, v$ are independent.
As the alignment is monotone, the event that all matches are correct, happens with probability $ 2^{-m} $.
The number of different $i,j$ pair is at most $n^2$.
The number of different size $m$ alignments is at most ${n'\choose m}\cdot {|w'|\choose m}  $.
There are $ 2^{3k'}$ different triples $(w, u, v)$.
And there are at most $2n'\times 2d' \le 4n'd'$ different $w'$, i.e. only consider $w'$ with length in $[n'-d', n'+d']$ is sufficient.
Hence, 
\begin{align*}
&\Pr[\text{\cref{prop:propertymainbasic} with parameter $d'$ is not satisfied}] \\
\le & {n'\choose m}\cdot {|w'|\choose m} \cdot 2^{3 k' } 4n'd' n^2  2^{-m}  \\
= & {n'\choose n'-m}\cdot {|w'|\choose |w'|-m} \cdot 2^{3 k'} 4n'd' n^2 2^{-m} \\
\le & {n'\choose d'}\cdot {n'+d' \choose 2d'}  \cdot 2^{3(1/3-\gamma)n'} 4n'd' n^2  2^{-(n'-d')}  && \text{ using } m \ge  (n'+|w'| - d')/2, |w'| \in [ n'-d', n'+d'] \\
\le &  {n'+d' \choose 2d'}^2   \cdot 2^{3(1/3-\gamma)n'} 4n'd' n^2  2^{-(n'-d')}  && \text{ using } {n'\choose d'} \le   {n'+d' \choose 2d'}\\
\le & 2^{4d'\log \frac{e(n'+d')}{2d'}}  \cdot 2^{ -3\gamma n' +d'} 4n'd' n^2   \\
\le & \frac{1}{n^{O(1)}}. && \text{ using } d' =   \frac{\gamma}{ 10 \log \frac{1}{\gamma}}  \cdot n' , n' = \frac{10\log n}{\gamma}  \\
\end{align*}

Thus, the outcome of the random generating process can meet \cref{prop:propertymainbasic} with a high probability. 

Next we replace uniform randomness by $\eps$-biased distribution $r $ over $\{0, 1\}^{k'n'n}$ from \cref{thm:epsbiased}, with $\eps = 2^{-3m}$, since $m = O(n') = O(\log n)$.
Notice that the seed length is $O(\log n)$.
In the above proof,  we consider a fixed alignment $M$ between $w$ and $ w'$.
Now we focus on the  bits in pairs of $M$.
Denote them as a string $s$.
Notice $|s| \leq 2m$.
Each   bit of $s$ is linear in bits of $r$.
So an XOR of any their subset is $\eps$-biased from uniform.
By \cref{lem:XORlem}, $s$ is $\eps \cdot \sqrt{2^{2m}} = 2^{-2m}$-close to uniform.
So $M$ has all its matches being correct with probability $\leq 2^{-m}+ 2^{-2m}$ instead of $2^{-m}$. 
The above analysis can still go through when this is the case.
So  we can try all seeds and find out a desired outcome.
\end{proof}

\subsection{Decoding and rate}

The decoding is via dynamic programming. Let $N = n\cdot n'$. Let $x\in \alphabetout^k$ be the original message, $z  =z^1\circ z^2 \circ \cdots \circ z^n = \C(x)$, where $z^i$ is a codeword of $\Cin^i$ for each $1\le i \le n$. 

Let $\corrupted$ be a corrupted codeword. For any substring $\corrupted_{[j:k]}$ of $\corrupted$, we say it can be matched to $i$-th block if $\ed(\corrupted_{[j:k]}, u)\le d'/2$ for some $u\in \Cin^i \setminus  \{0^{n'}\}$. 
We can define the following functions $g$ to denote whether a string $v$ can be matched to the $i$-th block, i.e. for any $v\in \alphabetin^*$ and $i\in [n]$, we let
\[g(i, v) = \begin{cases}
1, \text{ if } \exists \ u \in \Cin^i \setminus \{0^{n'}\}, \text{ s.t. }\ed(u,v) \le d'/2, \\
0, \text{ otherwise.}
\end{cases}
\]

For any $i\in [n]$ and $j\in [N]$, we say there is a non-zero alignment of size $m$ between $\corrupted_{[1:j]}$ and the first $i$ blocks if we can find $k$ non-overlapping substrings $\corrupted_{[\alpha_1: \beta_1]}, \dots,  \corrupted_{[\alpha_m: \beta_m]}$, $1\leq \alpha_1\leq \beta_1 \le \cdots \le \alpha_m \le \beta_m \le j$, and $m$ indices $1\le t_1< \dots< t_m\le i$ such that   $ g(t_l, \corrupted_{[\alpha_l: \beta_l]}) = 1$ for $l\in [m]$.

The dynamic programming computes the maximum size of   non-zero  alignments between  $\corrupted$ and the first $n$ blocks. It proceeds by computing a matrix $f$ such that $f[i,j]$ is the size of the largest alignment between $\corrupted_{[1:j]}$ and the first $i$ blocks. We let $f[i,j] = 0$ if $i = 0 $ or $j = 0$. For $1\le i \le n$ and $1\le j \le N$, we can compute $f[i,j]$ using the following updating rule:

\[f[i,j] = 
\max \begin{cases}
\underset{0\le j'<j}{\max} \; f[i-1, j'] + g(i, \corrupted_{[j'+1: j]}), \\
f[i, j-1].
\end{cases}
\]

Each value $f[i,j]$ corresponds an  alignment, denoted by $M[i,j]$. $M[i,j]$ has size $f[i,j]$ and can be represented as $M[i,j] = \left( (t_1,  [\alpha_1: \beta_1] ), \dots, (t_k,  [\alpha_k: \beta_k] )\right)$ where $k = f[i,j]$. Each tuple $(t_l, [\alpha_l: \beta_l] )$ represents a match  from the substring $\corrupted_{[\alpha_l: \beta_l]}$ to the $t_l$-th inner code. 
Thus, for any $l\in [k]$, we have $g(t_l,  y_{[\alpha_l: \beta_l]} ) = 1$. 

Specifically, $M[i,j]$ is constructed as following. For $i=0$ or $j=0$, $M[i,j]$ is empty. For $i,j$ both nonzero, if $f[i,j] = f[i, j-1]$, then $M[i,j] = M[i,j-1]$. Otherwise, $f[i,j] = f[i-1, j'] + g(i, \corrupted_{[j'+1, j]})$ for some $0 \le j' < j$. If there are multiple such $j'$, we choose an arbitrary one. Then, if $g(i, \corrupted_{[j'+1: j]}) = 1$, we append $(i, [j'+1: j] )$ to $M[i,j']$ to get $M[i,j]$. If  $g(i, \corrupted_{[j'+1: j]}) = 0$, we let $M[i,j] = M[i,j']$.

Assume $M$ is an alignment returned by the algorithm.
Let $\tilde{z}$ be such that if $t$ appears in a match $(t, [\alpha: \beta])$ in  $M$ then   $\tilde{z}^t  $ is recovered according to the match, i.e. it is equal to  $u \in \Cin^i \setminus \{0^{n'}\}, \text{ s.t. }\ed(u, \corrupted_{[\alpha: \beta]}) \le d'/2$.
For $t \in [n]$ not appearing in any matches of $M$, let $\tilde{z}^t = 0^{n'}$.

Finally, from $\tilde{z}$ we can recover the symbols of the outer code, then we use the decoder of $\Cout$ to get $x$.

\subsection{Analysis}

\begin{claim}
\label{claim:alignment_size1}
Any non-zero alignment between $\corrupted_{[1:j]}$ and the first $i$ blocks has size at most $f[i,j]$.
\end{claim}

\begin{proof}
We prove this by induction on $i$ and $j$. The base case is $i=0$ or $j=0$, and we have $f[i,j] = 0$ by definition. The claim holds trivially. For fixed $0 < i_0 < n$ and $0 < j_0 < N$, we assume the claim holds for any $i,j$ such that $0\le i \le i_0$ and $0 \le j \le j_0$. We show that the claim also holds for $f[i_0 + 1, j_0]$ and $f[i_0, j_0 + 1]$. 

For $f[i_0+1, j_0]$, let $M'$ be an optimal alignment between $\corrupted_{[1:j_0]}$ and the first $i_0 + 1$ blocks. There are two cases. First, there is a match  $(i_0+1, \corrupted_{[j' + 1:j_0]})$ in $M'$ for some $j'$. Then $M'\setminus \{(i_0 + 1, \corrupted_{[j': j_0]})\}$ is a matching between $\corrupted_{[1:j']}$ and the first $i_0$ blocks.  Since  $\forall j', f[i_0, j']$ is the size of the largest alignment between $\corrupted_{[1:j']}$ and the first $i_0$ blocks by the inductive hypothesis, we know $  f[i_0, j'] + 1 = |M'| $.
Thus $f[i_0+1, j_0] \ge f[i_0, j'] + 1 =  M'$. 
The second case is when there is no matching between any substring to $(i_0 + 1)$-th block in $M'$. Then $M'$ is an alignment between $\corrupted_{[1: j_0]}$ to the first $i_0$ blocks. By the inductive hypothesis, we have $|M'| =  f[i_0, j_0]$. By the updating rule of $f$, we have $f[i_0 + 1, j_0]\ge f[i_0, j_0]$.
So $f[i_0 + 1, j_0]\ge |M'|$.

The proof is similar for $f[i_0, j_0 + 1]$. 

\end{proof}

In the following discussion, we assume $\ed(z, \corrupted) = T$.
We first show the following claim. 
Let $\hat{n}$ be the number of non-zero blocks in $z$.
\begin{claim}
\label{claim:alignment_size2}
$f[n, N] \ge \hat{n}-2\frac{T}{d'}$.
\end{claim}
\begin{proof}
We show that there is an alignment between $\corrupted$ and the $n$ blocks with size $\hat{n}-2\frac{T}{d'}$. Then the claim is a direct result of \cref{claim:alignment_size1}. 

Since $\ed(z, \corrupted) = T$, there is a sequence of $T$ edit operations that transforms $z$ to $\corrupted$.   Denote $z = z^1\circ \cdots \circ z^n$, where $z_i \in \Cin^i, i\in [n]$. The $T$ edit operations can be partitioned into $n$ groups such that the $i$-th group of the edit operations is performed on $z^i$.  Notice that there are at most $2\frac{T}{d'}$ groups with more than $d'/2$ edit operations. Thus, for the remaining blocks, $\ed(z^i, \corrupted^i)\le d'/2$. For these blocks, if $z_i$ is not zero, then $\corrupted^i$ can be matched to $i$-th block. This gives an alignment of size at least $\hat{n}-2\frac{T}{d'}$.
\end{proof}

For any match $(t,  [\alpha:\beta] )$, we say it gives a correct recovery if $\tilde{z}^i = z^i$. Otherwise, we say it gives a wrong recovery. 
The alignment returned by the dynamic programming may contain wrong matches. The following lemma ensures that we can find a good amount of correct matches. 
\begin{lemma}
\label{lem:correctmatches}
The number of correct recoveries giving by matches in $M[n, N]$ is at least $\hat{n} - 8\frac{T}{d'}$.
\end{lemma}
\begin{proof}
We assume $M[n,N] $ has size $m$ i.e.
\[
M[n,N] = \left( (t_1,  [\alpha_1: \beta_1] ),(t_2,  [\alpha_2: \beta_2] ), \dots, (t_m, \ [\alpha_m: \beta_m] )\right).
\]

In the following analysis, we fix a set of $T$ edit operations that transform $z$ to $\corrupted$. 
We upper bound the number of wrong matches.
Consider a match $(i, [\alpha: \beta])$.
\begin{claim}
If there are $<d'/6$ errors in $ y_{[\alpha: \beta]}$, then  $\tilde{z}^i = z^i$.
\end{claim}
\begin{proof}
Suppose not.

Assume $ y_{[\alpha: \beta]}$ is achieved by applying $<d'/6$ insdels on a substring $ w $ of $z$.
the algorithm finds a $u\in \Cin^i$ s.t. $\ed(u, y_{[\alpha: \beta]} ) \le d'/2$, so $\ed(u, w) \le 2d'/3$.
Note that this implies $|w| \le n'+ \frac{2}{3}d'$.
Also note that $w$ can have overlaps with at most $3$ blocks of $z$, since otherwise its length is at least $2n'$.
Now one can pick a substring of $w'$ of $w$ which only has overlaps with  at most $2$ blocks of $z$, and $\ed(u, w') \le d'$.
Our inner codes have  \cref{prop:propertymainbasic}.
So $u$ has to be equal to the corresponding block of $z$ i.e. $z^i$.
And the algorithm set $\tilde{z}^i  = u$, hence the claim holds.
\end{proof}

Thus wrong recoveries i.e. $\tilde{z}^i \neq z^i$ by a match  $(i, [\alpha: \beta])$ must be s.t. there are at least $d'/6$ insdels in   $ y_{[\alpha: \beta]}$.
There are  $T$ errors in total.
So there can be at most $T/(d'/6) = 6T/d'$ wrong recoveries since  $[\alpha_i: \beta_i], i\in [m]$ do not overlap.

By \cref{claim:alignment_size2}, $M$ has size $\ge \hat{n} - 2T/d'$.
So the number of correct recoveries is at least $\hat{n} - 8T/d'$.
\end{proof}

\begin{lemma}\label{lem:transfertohamingdist}
$ \hamming(\tilde{z}, z) \le 16T/d'$.
\end{lemma}
\begin{proof}
By \cref{lem:correctmatches}, we can recover   at least $\hat{n} - 8T/d'$ blocks correctly.
Notice that these correctly recovered blocks are all non-zero since $M$ is a non-zero matching.
Now consider zero-blocks in $z$.
The number of them is $n-\hat{n}$.
For these blocks, the decoding algorithm makes mistakes in recovering them only when these wrong recoveries involve their indices.
Note that each wrong recovery only involves one index of  a zero block and $M$ is monotone.
So there are at most $8T/d'$ mistakes, i.e. at least $n-\hat{n} - 8T/d'$  zero blocks are correctly recovered.
Thus totally there are $\ge n - 16T/d'$ blocks correctly recovered.
Hence $ \hamming(\tilde{z}, z) \le 16T/d'$.
\end{proof}

\begin{proof}[Proof of \cref{thm:highrate1}]
Let the rate of the outer code be $1-\frac{\gamma}{2}$ and the rate of the inner code be $1/3-\gamma/2$.
By \cref{lem:innerconstruct}, we only need $d'/n' \le \frac{\gamma}{20\log \frac{2}{\gamma}}$ to get a sequence of inner code meeting \cref{prop:propertymainbasic}.
When $T \le \frac{\gamma}{64}nd'$, by \cref{lem:transfertohamingdist}, $\hamming(\tilde{z}, z) \le \gamma n/4$.
Since the relative distance of the outer code is $\gamma/2$.
One can decode correctly.

\end{proof}

%% file: construction_hrate2.tex
\section{For Rate Close to 1/2}
\label{sec:high rate 2}
In this section we modify our construction in the previous section to give a linear insdel code with rate close to $1/2$.

\begin{theorem}\label{thm:highrate}
For every constant $\gamma \in (0, 1/2)$, there exists a binary linear insdel codes with rate $\ge 1/2 - \gamma$, relative distance  $\rho \ge  \Omega(\gamma^3/\log \frac{1}{\gamma})$, and a polynomial time encoding.
\end{theorem}

\paragraph{The construction.} Following our general construction, here we take $\Cout$ to be an $[n,k,d]_{q}$ Reed-Solomon code with $|\alphabetout| = q =  \log n $, $k = (1- \delta) n$ and $d = \delta  n$, $\delta = \gamma /2$. 
We construct $n$ distinct binary inner codes $\Cin^1, \dots, \Cin^n$, each with codeword length $n'$, message length $k' =   \frac{1}{2}(1-\gamma') n' $  such that they satisfy the following \cref{prop:propertymain-s} with parameters $s = \frac{\delta t}{4}, t= O(\frac{\log\frac{1}{\gamma'}}{\gamma' \delta}), d' := \delta'n' = \Omega(\gamma' n'), \gamma' = \gamma/2$.
Recall the constructed code is denoted as $\C$.


\begin{property}[with parameters $s, t,  d'$]
\label{prop:propertymain-s}
$\forall i \in [n-t], j \in [n-t-1]$, $\forall w = w_i \circ \cdots \circ w_{i-1+t} \in \bigcirc_{l\in [t]} \Cin^{i-1+l}$,  $\forall u = u_j \circ u_{j+1} \circ \cdots \circ u_{j+t}\in \bigcirc_{l\in [t+1]}\Cin^{j-1+l}$, 
it always holds that for any substring $ w' $ of $u$, we have $\ed(w, w')\ge   d'$, 
as long as 
the number of unique blocks in $w$ or $u$   is  $\ge s$.

A unique block of $w$ (or $u$)  is defined to be a non-zero block such that there is no block  in $u$ (or $w$) that is the same codeword in a same inner code with this block.

We say the property is with parameter $(s, t, d')$ to indicate the parameters $s, t, d'$.
\end{property}

\begin{lemma}
\label{lem:propertymain-s}
For  every $\gamma' , \delta \in (0, 1/2)$,  
there exists an explicit construction of $n$ binary codes $\Cin^1, \dots, \Cin^n$ satisfying \cref{prop:propertymain-s} with parameter $(s, t, d')$, where each $\Cin^i$ has message rate $1/2 - \gamma' $, where $d'/n' = \frac{\gamma'} {10   }$, $t = \frac{10\log\frac{1}{\gamma'}}{\gamma'\delta}, s = \delta t/4$.
\end{lemma}
\begin{proof}
We first show that if the bases of the $n$ inner codes are randomly generated, then with high probability, the property holds.

For \cref{prop:propertymain-s},
suppose there exists a  substring $w'$ of $u$, such that $\ed(w, w') < d'$.
Then $\ell := \lcs(w, w') \ge (|w|+|w'| - d')/2 > tn' - d'$, since $|w'| > |w|-d'$, $|w| = tn'$.

Let $M$ be a monotone alignment between $w, w'$, with size $\ell$.
There are two possible cases for $M$.
Let $(i, j)$ denote one pair in $M$ where $i, j\in [n]$ are indices indicating the positions of the two entries in the two   codewords  in $\C$. 

Case 1, there exists $(i,j)$ in $M$ s.t. $|i-j| > d'$. 
Notice that in this case, it is not possible that there is $(i', j') \in M$ s.t. $i' = j'$, since then the edit distance is larger than $d'$.
Either of $w$ and $u$ has a number of non-zero blocks.
Let  $\hat{t} $ be the larger one of them.
The number of pairs in $M$ that connect to non-zero blocks is at least $\hat{t} n' - n'- d'$, since the worst case is that $\hat{t}$ is the number of non-zero blocks in $u$ and there are at most $n'+d'$ symbols of $u$   not matched.
As a result, the probability that $M$ is indeed matching same symbols happens with probability $ \le 2^{-(\hat{t} n' -n' -d')}$, because $M$ is monotone and bases of the inner codes are independently uniform at random.

The number of starting positions of the two intervals is at most $O(n^2)$.
The number of such alignments is at most ${tn' \choose \ell}{(t+1)n' \choose \ell}  $.
The number of positions of non-zero blocks is at most ${t \choose \hat{t}}{t+1 \choose \hat{t}} \le 2^{2t+2}$.
There are at most $2\hat{t}$  unique blocks in $w$ and $u$.
The number of choices for these non-zero blocks is at most $2^{ 2\hat{t} k'}$. 
So
\begin{align*}
    & 2^{-(\hat{t} n' -n' -d')} \times O(n^2) {tn' \choose \ell}{(t+1)n' \choose \ell}   2^{2t+2} 2^{2\hat{t}k'} \\
\le & 2^{-(\hat{t} n' -n' -d')} \times O(n^2) {tn' \choose d'}{(t+1)n' \choose d'}  2^{2t+2} 2^{2\hat{t}k'} && \text{ using } \ell \ge tn' - d', d' \le n'/2\\
\le & 2^{-(\hat{t} n' -n' -d')} \times O(n^2){(t+1)n' \choose d'}^2 2^{2t+2} 2^{2\hat{t}k'} && \text{ using } {tn' \choose d'} \le {(t+1)n' \choose d'}\\
\le & 2^{-(\hat{t} n' -n' -d')} \times O(n^2) \left(\frac{e(t+1)n'}{d'}\right)^{2d'} 2^{2t+2} 2^{2\hat{t}k'}  \\
\le & 2^{-\hat{t}(n'-2k') +n' + d' + 2t+2 } \times O(n^2) 2^{2d' \log \frac{e(t+1)n'}{d'}}  \\
\le & 2^{-\hat{t} 2\gamma' n' + n'+ 2t+2  + n'\log (2e(t+1))} O(n^2) && \text{ using } k' = (1/2-\gamma')n' ,  d'\le n'/2\\
\le & 1/n^{O(1)}. && \text{ using } \hat{t} \ge s , s = \delta t/4, t = \frac{10 \log \gamma'^{-1}}{\gamma' \delta},  \delta = \gamma' 
\end{align*}

Case 2, for every $(i,j) $ in $M$, $|i-j| \le d'$.
Notice that either $w$ and $u$ has a number of unique blocks.
Let $s'$ be the larger one.
Hence $s' \ge s$.
Without loss of generality, we do our proof assuming $u$ has a more number of unique blocks.
One can do a similar proof for the other condition.
Now we only consider matches s.t. one end is in a unique block of $u$ and the other end is in the corresponding block of $w$. 
The number of such matches is at least $ s'n' - n'-d' - s'd' =s'(n'-d') -n'- d' $.
This is because there are at most $n'+d'$ positions in unique blocks of $u$ that are not in any match, otherwise edit distance is larger than $d'$. 
And the term $s'(n'-d')$ comes from the following reason.
We claim that for each unique block there are at most $d'$ symbols matched to symbols in non-unique blocks.
Because otherwise  by the condition of this case, this block has to have matches to the left and right blocks of its corresponding block in the other string. There are totally $>d'$ such matches. Hence the corresponding block of the other string has to have $>d'$ symbols not matched.
Hence this shows the claim.
As a result, the probability that $M$ is indeed a matching happens with probability $\le 2^{-s'(n'-d')+ n'+ d' } $.
The number of starting positions of the two intervals is at most $O(nd')$.
The number of such alignments is at most ${tn' \choose \ell}{(t+1)n' \choose \ell}$.
So the number of positions of unique blocks is at most ${t \choose s'}{t+1 \choose s'} \le 2^{2t+2}$.
The number of different codewords in unique blocks is at most $2^{2s'k' }$.

So the probability that there is such a matching is at most
\begin{align*}
    & 2^{-s'(n'-d')+ n' + d' } \times O(nd') {tn' \choose \ell}{(t+1)n' \choose \ell} 2^{2t+2} 2^{ 2sk'} \\
\le & 2^{-s'(n'-d')+ n'+ d' } \times O(nd') {(t+1)n' \choose \ell}^2 2^{2t+2} 2^{ 2sk'} &&  {tn' \choose \ell} \le{(t+1)n' \choose \ell},  {t \choose s'} \le {t+1 \choose s'+1}\\
\le & 2^{ -2\gamma' n' s' + n'  + (s'+1)d'} \times O(nd') {(t+1)n' \choose \ell}^2 2^{2t+2}  && \text{ using } k' = (1/2 - \gamma') n'\\
\le & 2^{ -2\gamma' n' s' + n' + (s'+1)d'} \times O(nd') {(t+1)n' \choose n'+d'}^2 2^{2t+2}  && \text{ using } \ell \ge tn'-d'\\
\le & 2^{ -2\gamma' n' s'  + n'   + (s'+1)d'} \times O(nd')\left( \frac{e(t+1)n'}{n'+d'}\right)^2 2^{2t+2}  \\
\le & 2^{ -2\gamma' n' s' + n'  + (s'+1)d'} \times O(nd') 2^{2 n'\log \frac{e(t+1)}{1+\delta'}  } 2^{2t+2}  \\
\le & 1/n^{O(1)}. && \text{ using } d' = \frac{\gamma'}{10}n',   t = \frac{10\log\frac{1}{\gamma'}}{\gamma' \delta}, s = \delta t/4 
\end{align*}

As a result, there exists a sequence of inner codes meeting \cref{prop:propertymain-s}.

Next we derandomize the construction.
We generate the sequence of base vectors of all inner codes using an $\eps$-biased generator. 
Notice that in the above analysis the two interval only involve $O(\log n)$ codeword bits and they are linear of the base vectors. 
So any XOR of distinct bits from the two interval is $ \eps$-biased to a uniform random bit,
where by saying distinct bits we mean bits that are in distinct blocks or at distinct positions in a same code block.
By the XOR lemma, these distinct bits are   $\sqrt{2^{(2t+1)n'}} \eps$-close to uniform as long as $\eps$ is a small enough reverse polynomial.
So the above analysis can still go through.
The seed length is $O(\log \frac{n}{\eps})$.
Notice that the property can be checked in polynomial time.
So we can try every seed to find one meeting the property in polynomial time.

\end{proof}

We need the following dense property.
\begin{lemma}
\label{lem:dense}
Let $y, y'$ be two distinct codewords of the outer code, then for every $\ell \le n$, if we cut $y, y'$ into   non-overlapping substrings of length $\ell$, then for at least a $\frac{d}{2n}$ fraction of these substrings $y_t, t\in [\ell]$,   the distance between $y_t$ and $y'_{t}$ is at least  $ \frac{d}{2n} \ell $.
\end{lemma}
\begin{proof}
Suppose not. 
There are less than $\frac{d}{2n}$ fraction of these substrings $y_t, t\in [\ell]$,   the hamming distance between $y_t$ and $y'_{t}$ is at least  $ \frac{d}{2n} \ell $.
Then the total distance is at most $\frac{d}{2n} \times n + (1-\frac{d}{2n}) \times n \times \frac{d}{2n} < d $, which contradicts the distance of the outer code.

\end{proof}

Next we show the rate and the distance.
\begin{lemma}
\label{lem: distance and high rate}
  If \cref{prop:propertymain-s} holds for $ \Cin^i,i\in [n] $, then relative distance  is $  \frac{\delta'\delta  }{4t}$ and the rate is  $(1-\delta) (\frac{1}{2}-\gamma')  $.
\end{lemma}
\begin{proof}
    Consider two distinct messages $x, x'$.
    By \cref{lem:dense}, if we cut the codewords $C(x), C(x')$  into intervals each of length $tn'$, then for a $\ge \delta/2$ fraction of intervals, the distinct blocks (i.e. their corresponding outer codeword symbols are distinct) in the two corresponding intervals is at least $\delta/2$ fraction.
    We call these intervals of $C(x)$ as dense intervals. 
      Notice the number of distinct blocks of dense interval $i$ is at least $2s = \delta t/2 $.
    So either the $i$-th interval of $C(x)$ or the $i$-th interval of $C(x')$ has $\ge s = \delta t/4$ unique blocks.
    So either $C(x)$ or $C(x')$ has a $\ge \delta/4$ fraction of intervals such that each interval has $\ge s$ unique blocks.
    Without loss of generality assume it is $C(x)$.
    
    Let $M$ be a maximum matching between $C(x), C(x')$.
    For each dense interval $i$ of $C(x)$, consider matches of $M$ that involves symbols in this interval, denoted as $M_i$.
    Consider the   shortest string  of $ C(x')$ that cover $M_i$.
    By \cref{prop:propertymain-s},  the distance between them is $\ge d'$.
    So the overall relative distance between $C(x), C(x')$ is $\ge \frac{d'}{n'} \frac{\delta }{4t} =   \frac{\delta'\delta  }{4t}$.
    Also by our picking of parameters, the rate is $\ge  (1-\delta)  (\frac{1}{2} - \gamma')  $.

\end{proof}

The theorem follows immediately from our parameters.
\begin{proof}[Proof of \cref{thm:highrate}]
By \cref{lem:propertymain-s}, there is an explicit construction of $\Cin^i, i\in [n]$, with $\delta = \gamma / 2, \gamma' = \gamma/2$. 
So the encoding of $\C$ is in polynomial time.
By \cref{lem: distance and high rate}, the relative distance is $ \ge \frac{\delta' \delta }{4t} = \Omega( \gamma^3/\log \frac{1}{\gamma}  )$.
The rate is  $\ge (1-\delta)  (\frac{1}{2} - \gamma) \ge\frac{1}{2} - \gamma $.

\end{proof}

\subsection{Efficient decoding}

Next we show an efficient decoding for our code.
\begin{theorem}
\label{thm: efficient decoding of high rate}
The code given in \cref{thm:highrate} has a polynomial-time decoding with relative decoding radius $\Omega( \gamma^3/\log \frac{1}{\gamma}  ) $.
\end{theorem}

Let $\corrupted$ be a corrupted codeword with   at most $ \eta N$ insdel errors.

We  first define the following functions $g$.
For every $v\in \alphabetin^*$ and $i \in [n]$, we let
\[g(i, v) = \begin{cases}
1, \text{ if } \exists u \in \bigcirc_{j=i}^{i+t-1}\Cin^j  \text{ s.t. } \ed(u, v) < d'/2 ; \\
0, \text{ otherwise.}
\end{cases}
\]
We call such a match $( i, v)$ with $g(i, v) = 1$, as a segment match.
We call the $u \in \bigcirc_{j=i}^{i+t-1}\Cin^j $ s.t. $\ed(u, v) < d'/2 $, a candidate (string).
Denote the sequence of non-overlapping segment matches as a segment matching.
Here we only consider monotone matching i.e. substrings involved in    matches of the matching, are non-overlapping.
Let $f[i, j]$ be the size of the longest segment matchings between the first $i$ blocks and the first $j$ symbols of $\corrupted$.

We use the following dynamic programming to compute $f[i, j]$.
Let $f[i,j] = 0$ if $i = 0 $ or $j = 0$. 
For $1\le i \le n$ and $1\le j \le N$, we can compute $f[i,j]$ using the following transition rule:
\[f[i,j] = 
\underset{0\le j'<j}{\max} \; \{ f[i-t, j'] + g(i-t+1, \corrupted_{[j'+1: j]})\}.
\]
One can compute the corresponding matching along with the computing of $f[i, j]$.
Denote this matching as $M$.

Next, the decoder recovers all the blocks according to $M$.
Then assign zeros to all the other blocks.

Recall that $\corrupted$ is a corrupted version of $z$.
Denote $z_0$ as the outer codeword of the original message $x$.
\begin{lemma}
\label{lem:ellLB}
$|M| \ge \frac{n}{t} - \frac{2\eta N}{d' } $.
 
\end{lemma}

\begin{proof}

We cut $u$ into non-overlapping intervals of length $t$. 
There can be at most $\frac{\eta N'}{d'/2}$ intervals each having $\ge d'/2$ insdels.
Consider the canonical matching which matches every original symbol to its original position.
For those intervals which has less than $d'/2$ errors, it's distance to its original blocks is less than $d'/2$.
There are at least $ \frac{n}{t} - \frac{\eta N}{d'/2}$ such intervals.
So the claim holds.

\end{proof}

\begin{lemma}
\label{lem:DPgood}
    The DP can compute the maximum segment matching.
\end{lemma}
\begin{proof}
We use induction.

For the base case that either $i = 0 $ or $j = 0$, then the maximum matching should have size $0$.
So each $f[i, j]$ is the maximum for these cases.

For the induction, assume $f[i,j]$ is the maximum for every $i< i', j<j'$.
Then for $f[i', j']$, assume the corresponding maximum matching $M^*$ has its last segment match in between $\corrupted[j'', j']$ and blocks $[i'', i']$.
By the algorithm, $f[i' , j'] \ge 1 + f[i''-1, j''-1]$.
By the hypothesis, $f[i''-1, j''-1] \ge |M^*| - 1$.
So $f[i', j'] \ge |M^*|$.
This shows the induction.
\end{proof}

\begin{lemma}
\label{lem:efficient decoding correct}
The decoding can recover from $\eta$ fraction of errors for $\eta = \Omega(\gamma^3/\log \frac{1}{\gamma})$.
\end{lemma}
\begin{proof}
By \cref{lem:ellLB} and \cref{lem:DPgood}, we can get a matching $M$ of size  $\ge \frac{n}{t} - \frac{2\eta N}{d' }$.
There can be at most $\frac{2\eta N'}{d' }$ intervals each having $\ge d'/2$ errors.
So there are $\ell \ge \frac{n}{t} - 4\frac{\eta N}{d' }$ segment matches that each of them   involves $< d'/2$ errors.
For every such segment match $([i', i], v)$, assume the candidate is $u$. 
Assume $v$ is within insdel distance $d'/2$ to a substring $v'$ of $z$.
Assume the consecutive blocks of $z$ that contains $v'$ is $\tilde{v}$. 
We claim that there can be   $<s$ unique blocks in $u$ and $<s$ unique blocks in   $\tilde{v}$, comparing to each other.
This is because if not, then the insdel distance between $u$ and $v'$ is at least $d'$ by \cref{prop:propertymain-s}.
There are $<d'/2$ errors in $v$, so $([i', i], v)$ cannot be a segment match.
Hence the decoding can recover $>(t-2s)$ blocks in $\tilde{v}$.
So the total number of blocks that be correctly recovered is $> (t-2s) \ell \ge (t-2s)\left( \frac{n}{t} - 4\frac{\eta N}{d' }  \right)$.
By \cref{thm:listdecoding}, the relative list-decoding radius of Reed-Solomon code is $  1-\sqrt{1-\delta}  $.
So as long as  $(t-2s)\left( \frac{n}{t} - 4\frac{\eta N}{d' }  \right) \ge n \sqrt{1-\delta}   $, we can use the list decoding of the outer code to recover the message.
This deduces that $\eta \le ( 1- \frac{\sqrt{1-\delta}}{1-\delta/2}) \frac{\delta' }{4t} \le \frac{\delta \delta'}{8t} $.
So we can do unique decoding for $\eta = \Omega(\gamma^3/\log \frac{1}{\gamma}) $.

\end{proof}

\begin{proof}[Proof of \cref{thm: efficient decoding of high rate}]
The running time is polynomial, which immediately follows from our algorithm.
The decoding radius follows from \cref{lem:efficient decoding correct} .
\end{proof}

%% file: existence.tex
\section{The Existence of Linear Concatenated Codes matching Random Linear Codes}\label{section:existence}
In this section we show that there exist linear concatenated insdel codes that match the parameters of random linear codes.

We will use the following parameters in our general construction: the outer code is taken to be an $[n, k=(1-\gamma)n/2, d=(1+\gamma) n/2]_q$ Reed-Solomon code with $q=\Theta(n)$, and every inner code is an independent random linear code over the field $\F_{q_0}$ with block length $n' = \log_{q_0} q $ and message length $k'=\log_{q_0} q$. That is, every inner code has rate $1$.

We'll need the following claim.

\begin{claim}\label{claim:RS} \cite{1056765}
  For any $[n, k, d]_q$ Reed-Solomon code, the number of codewords with weight $w$ is at most ${n \choose w}q^{w-d+1}$.
\end{claim}

We now have the following lemma.

\begin{lemma}

In the general construction, if the outer code is an $[n, k=(1-\gamma)n/2, d=(1+\gamma) n/2]_q$ Reed-Solomon code with $q=\Theta(n)$, $\gamma=\delta +2\ent(\delta)/\log q_0+o(1)$, and every inner code is an independent random binary linear code with block length $n' = \log_{q_0} q $ and message length $k'=\log_{q_0} q$, then with probability $\ge 3/4$, the final code can correct $\delta$ fraction of insdel errors.
\end{lemma}

\begin{proof}
    Let $N=nn'$. The concatenated code can correct $\delta$ fraction of insdel errors unless there exist two different codewords whose longest common subsequence has length $\geq \ell=(1-\delta)N$. We bound the probability that this event happens, as follows.

    First, The total number of length $\ell$ monotone matchings is ${N \choose \ell}^2$. Now we bound the total number of pairs of codewords, and the probability that the event happens, using a set of parameters we define. Trivially, the number of pairs of codewords is bounded by $q^{2k}$. However, we can bound this number in a different way: for any pair of outer codewords $(C_1, C_2)$, we first consider $C_1-C_2$, which by linearity is also an outer codeword. Let $w \geq d$ denote the weight of $C_1-C_2$. By Claim~\ref{claim:RS} the number of such codewords is at most ${n \choose w}q^{w-d+1}$. Note that when $C_1-C_2$ is fixed, the pair $(C_1, C_2)$ is fixed by fixing $C_1$.

Now for any monotone matching between the symbols of the final concatenated codewords corresponding to $C_1, C_2$, we define two parameters $s, t$. Below when we say a ``block", we mean a block of symbol positions which correspond to an inner codeword.
\begin{itemize}
    \item $s$ is the number of blocks that contain at least one match whose two endpoints are both in this block, and that the two symbols in $C_1, C_2$ corresponding to this block are the same.
    \item We say a $0$ symbol in $C_1$ or $C_2$ is unique if the symbol at the same position of the other codeword is not $0$. Let $t$ be the total number of unique $0$ symbols in $C_1, C_2$.
\end{itemize}
Notice that $s \le n-w$, since there are at most $n-w$ symbols that are the same and at the same positions in $C_1, C_2$. Also notice that $t \le 2k$ since each of $C_1, C_2$ can have at most $n-d=k$ of such $0$ symbols.

Using these parameters, we can bound the probability that such a monotone matching happens. First, we remove all matches whose two endpoints are both in one same block and the two symbols in $C_1, C_2$ corresponding to this block are the same. Note that the number of such matches is at most $sn'$. Now, all remaining matches must have the two endpoints in different blocks, or in the same block but the two symbols in $C_1, C_2$ corresponding to this block are different. By our choice of independent random linear inner codes, and using a similar argument as in Section~\ref{sec:higherror}, if we count these matches from left to right, then every such match happens with probability $1/q_0$ conditioned on all previous matches, unless such a match hits two $0$ symbols in $C_1, C_2$ (in which case we say the match is \emph{bad}). However, the number of such bad matches is at most $tn'/2$. Therefore, the probability that the matching happens is at most 
\[q_0^{-(\ell-sn'-tn'/2)} \leq q_0^{-(w-\delta n-t/2)n'}.\]

On the other hand, note that the $t$ unique $0$ symbols can only happen at positions where the symbols of $C_1$ and $C_2$ differ. Since $C_1-C_2$ is fixed, this also fixes $t$ symbols in $C_1$. Hence the total number of possible codewords for $C_1$ is at most $q^{k-\min(k, t)} \leq q^{k-t/2}$ since $t \leq 2k$. Now we bound the probability that there exist two different codewords whose longest common subsequence has length $\geq \ell=(1-\delta)N$.

\begin{align}
     & \Pr[\exists \mbox{ two different codewords whose longest common subsequence has length at least } \ell ] \\
 \le & \sum_{w\ge d, s, t} {nn' \choose \ell}^2 {n \choose w}q^{w-d+1} {n-w \choose s} {w \choose t}2^t q^{k-t/2} q_0^{-(w-\delta n-t/2)n'} \\
 = & \sum_{w\ge d, s, t} {nn' \choose \delta nn'}^2 {n \choose w}  {n-w \choose s} {w \choose t}2^t q_0^{(k-t/2+w-d+1-(w-\delta n-t/2))n'} \\
 = & \sum_{w\ge d, s, t} {nn' \choose \delta nn'}^2 {n \choose w}  {n-w \choose s} {w \choose t}2^t q_0^{(k-d+1+\delta n)n'} \\
 \le & \sum_{w\ge d, s, t}2^{2nn' \ent(\delta)  }  {n \choose w}{n-w \choose s} {w \choose t}2^t q_0^{(1-\gamma n+\delta n)n'}\\ 
  \le & \sum_{w\ge d, s, t} 2^{2nn' \ent(\delta)  } 2^n 2^{n-w}2^w 2^t q_0^{(\delta-\gamma)nn'+n'}\\ 
  . \le & \sum_{w\ge d, s, t} 2^{2nn' \ent(\delta)} 2^{3n} q_0^{(\delta-\gamma)nn'+n'} \\
  \le & n^3 2^{2nn' \ent(\delta)} 2^{3n} q_0^{(\delta-\gamma)nn'+n'}.   
\end{align}

To make the above probability $ \leq 1/4$, one just needs
\[3 \log n+2nn' \ent(\delta)+3n+(\delta-\gamma)nn' \log q_0 +n' \log q_0 \leq -2.\]

Thus it suffices to take 

\[\gamma=\delta +2\ent(\delta)/\log q_0+(3n+3\log n+2)/(nn' \log q_0)+1/n=\delta +2\ent(\delta)/\log q_0+o(1).\]
\end{proof}

This gives the following theorem.

\begin{theorem}
For any field $\F_q$, there exists a linear concatenated code over $\F_q$ where the outer code is a Reed-Solomon code, such that the code has rate $\frac{1}{2}(1- \delta)- \ent(\delta)/\log q-o(1)$ and can correct $\delta$ fraction of insdel errors, where $\ent()$ is the binary entropy function. 
\end{theorem}

%% file: half_plus_logn_lower_bound.tex
\section{A Lower Bound}\label{sec:lowerbound}

    In this section, we give the lower bound that any linear subspace $\C \subseteq \F_2^n$ with dimension $\geq 3$ must contain two strings whose longest common subsequence has length at least $(\frac{1}{2}+\frac{3}{16\log n})n$. 
    
    We regard all binary strings of length $n$ as a linear subspace and define ``$+$'' as the bit-wise xor operation.
    
    \begin{theorem}\label{thm:lower bound main}
        For every three different strings $a,b,c$ of length $n$, the linear subspace consists of 8 strings $0^n,a,b,c,a+b,a+c,b+c,a+b+c$, we can always pick string $x,y$ out of these 8 strings, such that $\lcs(x,y) \geq \frac{n}{2} + \Omega(\frac{n}{\log n})$.
    \end{theorem}
    

    For simplicity, we slightly misuse the notation by labelling the three string we choose by $a,b,c$, their sum is labeled $a+b,a+c,b+c,a+b+c$, every $\C \subseteq \F_2^n$ with dimension $\geq 3$ must contain 8 different strings $0^n,a,b,c,a+b,a+c,b+c,a+b+c$.
    
    In the following proof, we need to compare characters of different length $n$ strings in the same position, so we define the matrix generated by strings as follows:

    \begin{definition}
        $m$ length-$n$ strings $a_1,a_2,...,a_m$ generates a $m$ by $n$ matrix $T$ by $T(i,j):=$ the $j$-th character of string $a_i$.
    \end{definition}


    \begin{lemma}\label{lem:toolthmprelem}
        If the largest pairwise $\lcs$ between $0^n,a,b,c,a+b,a+c,b+c,a+b+c \in \F_2^n$ has its size no larger than $\frac{n}{2} + \frac{3n}{16\log n}$, then there exist 3 strings $x,y,z$ out of $a,b,c,a+b,a+c,b+c,a+b+c$ such that 
        \begin{itemize}
            \item    $x + y = z$,
            \item  each of $x,y,z$ consists of $ [\frac{n}{2} - \frac{3n}{16\log n}, \frac{n}{2} + \frac{3n}{16\log n}]$ 0s,
            \item the $3$ by $n$ matrix generated by $x,y,z$  contains almost the same number of different columns, namely, column $(0,0,0)^T,(0,1,1)^T,(1,0,1)^T,(1,1,0)^T$ each appears $[\frac{n}{4} -\frac{9n}{16\log n}, \frac{n}{4} + \frac{9n}{16\log n}]$ times. 
        \end{itemize}

    \end{lemma}

    
   
   \begin{proof}
        Without loss of generality, we assume string $a$ has the most number of $1$s among $0^n,a,b,c,a+b,a+c,b+c,a+b+c$.
   
        When the $\lcs$ between any two of $0^n,a,b,c,a+b,a+c,b+c,a+b+c$ is no larger than $\frac{n}{2} + \frac{3n}{16\log n}$, if $a$ has more than $\frac{n}{2} + \frac{3n}{16\log n}$ number of $1$s, then every other string must have no more than $\frac{n}{2} + \frac{3n}{16\log n}$ number of $1$s, since otherwise, we can find a common subsequence of $1$s of length larger than $\frac{n}{2} + \frac{3n}{16\log n}$.

        So $b,c,a+b,a+c,b+c,a+b+c$ has no more than $\frac{n}{2} + \frac{3n}{16\log n}$ number of $1$s.
        
        We first generate a $7$ by $n$ matrix by $a,b,c,a+b,a+c,b+c,a+b+c$, and count the frequency of different columns as follows: 
        
        {
        \begin{table}[h]
            \centering
            \setlength{\extrarowheight}{2pt}
            \begin{tabular}{cc|c|c|c|c|c|c|}
                    & \multicolumn{1}{c}{} & \multicolumn{5}{c}{Column Counting Matrix $C$}\\
                    & \multicolumn{1}{c}{} & \multicolumn{1}{c}{$t_0$}  & \multicolumn{1}{c}{$t_1$} & \multicolumn{1}{c}{$t_2$} & \multicolumn{1}{c}{$t_3$} & \multicolumn{1}{c}{$t_4$} & \multicolumn{1}{c}{$t_5$} \\\cline{3-8}
                \multirow{3}*{}  
                    & $a$ & $0$ & $1$ & $0/1$ & $0/1$ & $1$ & $0$\\\cline{3-8}
                    & $b$ & $1$ & $0$ & $0$ & $1$ & $1$ & $0$ \\\cline{3-8}
                    & $c$ & $1$ & $0$ & $1$ & $0$ & $1$ & $0$ \\\cline{3-8}
                    & $a+b$ & $1$ & $1$ & $0/1$ & $1/0$ & $0$ & $0$ \\\cline{3-8}
                    & $a+c$ & $1$ & $1$ & $1/0$ & $0/1$ & $0$ & $0$ \\\cline{3-8}
                    & $b+c$ & $0$ & $0$ & $1$ & $1$ & $0$ & $0$ \\\cline{3-8}
                    & $a+b+c$ & $0$ & $1$ & $1/0$ & $1/0$ & $1$ & $0$ \\\cline{3-8}
            \end{tabular}
        \end{table}
        }
       
        We denote the frequency of column $(0,1,1,1,1,0,0)^T$ as $t_0$, the frequency of $(1,0,0,1,1,0,1)^T$ as $t_1$, 
        the total frequency of $(0,0,1,0,1,1,1)^T$ and $(1,0,1,1,0,1,0)^T$ as $t_2$, the total frequency of $(0,1,0,1,0,1,1)^T$ and $(1,1,0,0,1,1,0)^T$ combined as $t_3$, the frequency of $(1,1,1,0,0,0,1)^T$ as $t_4$, the frequency of $(0,0,0,0,0,0,0)^T$ as $t_5$, so $\sum_{i = 0}^5 t_i = 1$.

        We compute the number of 1$s$ of $b+c,b,c$, and get $t_2 + t_3 \leq \frac{1}{2}+ \frac{3}{16\log n}, t_0 + t_3 + t_4 \leq \frac{1}{2}+ \frac{3}{16\log n}, t_0 + t_2 + t_4 \leq \frac{1}{2}+ \frac{3}{16\log n}$, so $t_0 + \max\{t_2,t_3\} + t_4 \leq \frac{1}{2}+ \frac{3}{16\log n}$.

        Since we have an upper bound on the length of the longest common subsequence, we now find common subsequences of different strings to get inequalities.   
        
        Note that $b$ and $c$ have exactly $(t_0 + t_1 + t_4 + t_5)n$ equal characters in their corresponding positions, which corresponds to a length $(t_0 + t_1 + t_4 + t_5)n$ common subsequence, so by assumption we know that $t_0 + t_1 + t_4 + t_5 \leq \frac{1}{2} + \frac{3}{16\log n}$.
        Hence, $t_2 + t_3 \geq \frac{1}{2} - \frac{3}{16\log n}$.
        So $t_2 + t_3 \in [\frac{1}{2} - \frac{3}{16\log n},\frac{1}{2} + \frac{3}{16\log n}]$. 
        
        By counting the number of $1$s of $b$ and $c$, since $b$ has $(t_0 + t_3 + t_4)n$ $1$s and $c$ has $(t_0 + t_2 + t_4)n$ $1$s, we can find a common subsequence of $t_0 + t_4 + \min\{t_2,t_3\}$ $1$s, so $t_0 + t_4 + \min\{t_2,t_3\} \leq \frac{1}{2} + \frac{3}{16\log n}$.

        We compute the $\lcs$ between $b$ and all zero string, namely the number of $0$s in $b$, so we know that $t_1 + t_2 + t_5 \leq \frac{1}{2} + \frac{3}{16\log n}$, compute the $\lcs$ between $c$ and all zero string, we know that $t_1 + t_3 + t_5 \leq \frac{1}{2} + \frac{3}{16\log n}$, combined it together, we know $t_1 + \max\{t_2,t_3\} + t_5 \leq \frac{1}{2} + \frac{3}{16\log n}$, which equivalents to $t_0 + t_4 + \min\{t_2,t_3\} \geq \frac{1}{2} - \frac{3}{16\log n}$.

        So we have $t_0 + t_4 + \min\{t_2,t_3\} \in [\frac{1}{2} - \frac{3}{16\log n}, \frac{1}{2} + \frac{3}{16\log n}]$, and since $t_0 + \max\{t_2,t_3\} + t_4 \leq \frac{1}{2}+ \frac{3}{16\log n}$, we have $t_2 - t_3 \in [-\frac{3}{8\log n}, \frac{3}{8 \log n}]$.
      
        Since $t_2 + t_3 \in [\frac{1}{2} - \frac{3}{16\log n},\frac{1}{2} + \frac{3}{16\log n}]$, it holds that $t_2,t_3 \in [\frac{1}{4} - \frac{3}{8\log n}, \frac{1}{4} + \frac{3}{8\log n}]$.
        Together with inequalities above, we know $t_0 + t_4 \in [\frac{1}{4} - \frac{9}{16\log n}, \frac{1}{4} + \frac{9}{16\log n}], t_1 + t_5 \in [\frac{1}{4} - \frac{9}{16\log n}, \frac{1}{4} + \frac{9}{16\log n}]$.
        So consider $b,c,b+c$.
        The first bullet  is satisfied immediately.
        The second bullet immediately follows from the lemma condition.
        The argument above shows that they satisfy the third bullet.

   \end{proof}

   
   To further prove \cref{thm:lower bound main}, we show the following theorem.
   \begin{theorem}\label{thm:toolthm}
        For three length $n$ strings $x,y,z$, consider the $3$ by $n$ matrix  generated by them.
       If each of the columns $(0,0,0)^T,(0,1,1)^T,(1,0,1)^T,(1,1,0)^T$  appears $[\frac{n}{4} -\frac{9n}{16\log n}, \frac{n}{4} + \frac{9n}{16\log n}]$ times , then the maximum of  pairwise $\lcs$'s of $x,y,z$ is larger than $\frac{n}{2} + \frac{3n}{16\log n}$.
   \end{theorem}

    To prove this theorem, we need following definitions which defines the ``level structure'' of strings that is used to find common subsequences sufficiently long to provide desired the lower bound.


   \begin{definition}
        Given a string $L$ whose alphabet size is no more than 3, We add pairs of brackets,   to $L$ and define the following level structure with regard to $L$. Assume the alphabet set is  $\{\alpha,\beta,\gamma\}$.
    \begin{itemize}
        \item   A pair of brackets is called level $1$, if the substring of $L$ inside the brackets can be expressed as $\alpha^k \beta^k$ or $\alpha^k\gamma^k$ for some $k \geq 0$.
        \item A pair of brackets is called level $n$, if the following holds.
        \begin{itemize}
            \item Every pair of brackets inside has a level smaller than $n$;
            \item At least one of the pairs inside is level $n - 1$;
            \item Removing all those bracket pairs with level $<n$ and together with their contained strings, the remaining string inside the bracket can be expressed as $\alpha^k \beta^k$ or $\alpha^k\gamma^k$ for some $k \geq 0$.
        \end{itemize} 
        \item $L$ is called level $n$, if we can add pairs of brackets of levels less than $n$ to $\langle L\rangle$, such that the pair of out-most brackets is level $n$.

        \item   A pair of brackets is called enhanced level $1$, if the substring string of $L$ inside the brackets can be expressed as $\alpha^k \{\beta,\gamma\}^k$ for some $k \geq 0$.        
        \item A pair of brackets is called enhanced level $n$ if the following holds.
        \begin{itemize}
            \item All pairs of brackets inside have enhanced levels smaller than $n$;
            \item At least two inside pairs are enhanced level $n - 1$;
            \item Removing all pairs of brackets (with level $\le n-1$) together with their contained strings, the remaining string can be expressed as $y\circ z$, $y\in \alpha^k, z\in \{\beta,\gamma\}^k$ for some $k \geq 0$;
            \item The removed brackets are not in between symbols of $z$.
        \end{itemize} 
         We call $2k$ the inner length of this enhanced level $n$.
        \item A string $L$ is called enhanced level $n$ if we can add pairs of brackets of enhanced levels less than $n$ to $\langle L\rangle$, such that the pair of out-most brackets is enhanced level $n$.
        
        \end{itemize}

   \end{definition}

   We denote the column vector $(0,1,1)^T,(1,0,1)^T,(1,1,0)^T,(0,0,0)^T$ by $a,b,c,d$, so the $3$ by $n$ matrix we mentioned can be seen as a length $n$ string with character $a,b,c,d$, we ignore character $d$, and denote the remaining new string as $L$ while the string before the ignorance of $d$ as $L'$.

   \begin{lemma}
      One can add at most $n$ pairs of brackets to $L$ and this implies that the level of $L$ is   $O(n)$.
   \end{lemma}
   
   \begin{proof}
        By applying the following algorithm, we can add brackets to L that fits the definition of ``level structure'':
        
    \begin{algorithm}
    \begin{algorithmic}
        \State $current \gets$ the first character of $L$, $currentcount \gets 1, next \gets null$
        \State $nextcount \gets 0, leftbracket \gets 0, rightbracket \gets 0$
        \While{we haven't reach the last character of L}
            \State Read in next character, denoted as m.
            \State $rightbracket \gets rightbracket + 1$
            \If{$currentcount == nextcount$}
                \State $current \gets m, next \gets null, currentcount \gets 0, nextcount \gets 0$
                \State Add a left bracket to the position of $leftbracket$,Add a right bracket to $rightbracket$
                \State $leftbracket \gets rightbracket$ 
            \EndIf

            \If{$m == current$}
                \State $currentcount \gets currentcount + 1$
            \Else
                \If{$next == null$}
                    \State $next \gets m, nextcount \gets 1$
                \Else
                    \If{$m == next$}
                        \State $nextcount \gets nextcount + 1$
                    \Else
                        \State Add a right bracket to $rightbracket$
                        \State Add a left bracket to $leftbracket + currentcount - nextcount$
                        \State $next \gets m, currentcount \gets currentcount - nextcount, nextcount \gets 1$
                    \EndIf
                \EndIf
            \EndIf

        \EndWhile
   
    \end{algorithmic}
    \end{algorithm}
   Since we add at most 1 bracket for each loop, and the loop stops in $n$ iterations, we add at most $n$ brackets, so $L$ is level $O(n)$ with respect $\{a,b,c\}$, and contains at most $n$ brackets.
        
   \end{proof}

    \noindent{\bf Example 1}

    For string $s = aaaabbcacaaabbcc$,  by using above algorithm, we can add the brackets so as to have $$s = \langle a\langle a \langle aabb \rangle_1 c \rangle_2 \langle ac \rangle_1 a\langle aabb \rangle_1 cc \rangle_3.$$
    
    It shows that $s$ can be called a level 3 bracket with respect to $\{a,b,c\}$(the level of each bracket is marked in the corresponding right bracket).

   \begin{lemma}
      A level $n$ string with respect to $\{a, b, c\}$ and with $O(n)$ brackets is an enhanced level $\ell$ string, where $\ell \le O(\log n)$.
   \end{lemma}
   
   \begin{proof}

        By using the following algorithm, we can shift the level $n$ string with $O(n)$ brackets to enhanced level $O(\log n)$.
        
        We consider these pairs of brackets that are not contained in any other pair of brackets. 
        For every such pair, denote its level by $h = O(n)$. 
        For $k = h, h - 1, ... 2$, whenever inside a level $k$ pair $p$ of brackets there is only one lower level pair $q$ which contains all the other brackets inside $p$, we delete the bracket pair $q$.

        By definition of enhanced level $\ell$ strings, one can see that the above operation gives us an enhanced level $\ell$ string. 
        Notice that  if after the operation we get a string that is an enhanced level $\ell$ string, then the original string also has to be an enhanced level $\ell $ string.
        So we only need to show that $\ell = O(\log n)$.
        
        By definition, an enhanced level $\ell$ pair of brackets contains at least $2^h$ pairs of brackets. 
        Notice that the above algorithm provides us with an enhanced pair of brackets  that contains $O(n)$ brackets since we only delete brackets. 
        So  $\ell$ has to be $ O(\log n)$.

   \end{proof}

    \noindent{\bf Example 2}
   
    For string $s = aaaabbcacaaabbcc$, by using the above algorithm to transform the level 3 bracket, we get the partition $$s = \langle a \langle aaabbc\rangle_1 \langle ac\rangle_1 a \langle aabb \rangle_1 cc \rangle_2.$$ 
    It shows that $s$ can be called an enhanced level 2 bracket with respect to $\{a,b,c\}$(the level of each bracket is marked in the corresponding right bracket).

    A restricted common subsequence $s \in \{0, 1\}^t$ between two equal lengths strings in $\{0, 1\}^n$ for some $ t\le n$, is defined to be s.t. both strings have a subsequence $s$ and for each $s_i$, if $s_i = 1$, then the corresponding $1$ symbols in the two strings must have the same position.
    Recall that $L'$ is string of $\{a, b, c, d\}$ while each one of $a, b, c, d$ is actually a 3 dimension vector in $\{0, 1\}$. 
    So $L'$ actually represents  3 strings, i.e.  we can regard $L'$ as a $3$ by $n$ matrix, each row is a string. 
    Also, recall that the string $L$ is generated by using $L'$ and deleting the $d$ symbols inside.
    Let the sum of lengths of the three pairwise restricted common subsequences between the three strings represented by $L'$,  be $t + n$, where we define $t$ as the additional value of the string $L'$ or $L$.
   
    In what follows, we use the ``level structure'' of string $L$ to prove the additional value of $L$ is $\Omega(n/ \log n)$. 

    It follows from the definition that the substrings and brackets contained in an enhanced level $\ell$ bracket, has the form $a^{i_0}\langle...\rangle a^{i_1}\langle...\rangle...a^{i_m}\{b,c\}^{\sum_{j = 0}^m i_j}$, for some $m \geq 2$. In the following lemma, we will show that the additional value is larger than $\frac{3n}{16\log n}$ by showing we can either compute the additional value by ignoring all the inner brackets, or by taking the summation of the additional values of inner brackets that are not contained by any other inner brackets. 
    Note that the enhanced level of  inner brackets are no more than $\ell - 1$. 
    So we can repeat the process $\ell$ times and have a desired lower bound.
    We phrase the idea in detail as the following.

    \begin{lemma}\label{lem:addtional value}
         For any $\ell$, the additional value of 
         an enhance level $\ell$ string,
         is at least a quarter of the inner length or the sum of additional values of lower level bracket pairs which are only contained by the level $\ell$ bracket pair.
    \end{lemma}
    
    \begin{proof}
        By definition, the string inside an enhanced level $\ell$ bracket pair must be in the form $a^{i_0}\langle...\rangle a^{i_1}\langle...\rangle...a^{i_m}\{b,c\}^{\sum_{j = 0}^m i_j}$, for some $m \geq 2$.
        We denote the number of $a,b,c$ inside the brackets to be $a_1,b_1,c_1$, then $a_1 = b_1 + c_1$.
        Assuming the number of $b$ in $\{b,c\}^k$ to be $ b_2 $.
        Recall that a string here represents a $3$-row matrix where each row is a binary string.
        When we count the 0s between the first row and the second row, we have value $\max\{b_1,c_1\} + \max\{b_2,k - b_2\}$. 
        For the first row and the third row, we have value $\max\{b_1,c_1\} + \max\{b_2, k - b_2\}$.
        For the second and the third row, we have value $a_1 + k + \min\{b_2,k - b_2\}$. 
        One can sum them up and deduce that the additional value $t$ is at least $\max\{b_2, k - b_2\} \geq \frac{k}{2}$.
        Notice that the inner length is $2k$ by definition.
        
        On the other hand, consider a sequence of lower-level brackets that are contained by each other. 
        Notice that the additional value of the whole string is at least the summation of the additional values of  the strings covered by these lower-level brackets.
        So we get our proof.
         
    \end{proof}

    Since $L$ is a  string with enhanced level  $\ell = O(\log n)$ and the length of $L$ is   $v \in [\frac{3n}{4}-\frac{9n}{16 \log n}, \frac{3n}{4}+\frac{9n}{16 \log n}]$, we denote the sum of inner lengths of enhanced level $i$ brackets for $L$, as $a_i$.
    So $\sum_{i = 1}^{\ell} a_i = v$.
    \cref{lem:addtional value} shows that the additional value is no less than $\max_{i = 1,...,\ell} a_i \ge \frac{v}{\log n}$. 
    By the definition of additional value, we can find three restricted common subsequences with total length $\ge \frac{3n}{4} + v/\log n$.
    Recall that we deleted the columns $(0, 0, 0)^T$ in the matrix generated by $x,y,z$.
    The number of them is $n-v$.
    So we can construct our final common subsequences by adding back the corresponding zeros to the restricted common subsequences.
    Hence their total length is  $ \frac{3n}{4} + v/\log n +3(n-v)$.
    This implies that there exists a pair of $x,y,z$ such that the $\lcs$ of this pair is at least $\frac{n}{2}+\frac{3n}{16\log n} $.
    This shows \cref{thm:toolthm}.

    To show \cref{thm:lower bound main}, suppose it does not hold, i.e.  the largest pairwise $\lcs$ between $0^n,a,b,c,a+b,a+c,b+c,a+b+c \in \F_2^n$ to be no larger than $\frac{n}{2} + \frac{3n}{16\log n}$.
    Then by \cref{lem:toolthmprelem}, one can find the $x,y,z$ in \cref{thm:toolthm}.
    Hence by \cref{thm:toolthm},  there exists a pair of $x,y,z$ such that the $\lcs$ of this pair is at least $\frac{n}{2}+\frac{3n}{16\log n} $.
    This is a contradiction.

%% file: main.bbl
\newcommand{\etalchar}[1]{$^{#1}$}
\begin{thebibliography}{AGHP92}

\bibitem[AGFC07]{AGFC07}
Khaled~A.S. Abdel-Ghaffar, Hendrik~C. Ferreira, and Ling Cheng.
\newblock On linear and cyclic codes for correcting deletions.
\newblock In {\em 2007 IEEE International Symposium on Information Theory
  (ISIT)}, pages 851--855, 2007.

\bibitem[AGHP92]{alon1992simple}
Noga Alon, Oded Goldreich, Johan H{\aa}stad, and Ren{\'e} Peralta.
\newblock Simple constructions of almost k-wise independent random variables.
\newblock {\em Random Structures \& Algorithms}, 3(3):289--304, 1992.

\bibitem[BGH17]{BGH17}
Boris Bukh, Venkatesan Guruswami, and Johan H{\aa}stad.
\newblock An improved bound on the fraction of correctable deletions.
\newblock {\em {IEEE} Trans. Information Theory}, 63(1):93--103, 2017.
\newblock Preliminary version in SODA 2016.

\bibitem[BGZ18]{BGZ-journal}
Joshua Brakensiek, Venkatesan Guruswami, and Samuel Zbarsky.
\newblock Efficient low-redundancy codes for correcting multiple deletions.
\newblock {\em IEEE Transactions on Information Theory}, 64(5):3403--3410,
  2018.
\newblock Preliminary version in SODA 2016.

\bibitem[BLC{\etalchar{+}}16]{Bornholt16}
J.~Bornholt, R.~Lopez, D.~M. Carmean, L.~Ceze, G.~Seelig, and K.~Strauss.
\newblock A dna-based archival storage system.
\newblock {\em ACM SIGARCH Comput. Archit. News}, 44:637–--649, 2016.

\bibitem[CGHL21]{CGHL21}
Kuan Cheng, Venkatesan Guruswami, Bernhard Haeupler, and Xin Li.
\newblock Efficient linear and affine codes for correcting
  insertions/deletions.
\newblock In {\em Proceedings of the 2021 ACM-SIAM Symposium on Discrete
  Algorithms (SODA)}, pages 1--20, 2021.

\bibitem[CJLW19]{cheng_et_al:LIPIcs:2019:10613}
Kuan Cheng, Zhengzhong Jin, Xin Li, and Ke~Wu.
\newblock {Block Edit Errors with Transpositions: Deterministic Document
  Exchange Protocols and Almost Optimal Binary Codes}.
\newblock In Christel Baier, Ioannis Chatzigiannakis, Paola Flocchini, and
  Stefano Leonardi, editors, {\em 46th International Colloquium on Automata,
  Languages, and Programming (ICALP 2019)}, volume 132 of {\em Leibniz
  International Proceedings in Informatics (LIPIcs)}, pages 37:1--37:15,
  Dagstuhl, Germany, 2019. Schloss Dagstuhl--Leibniz-Zentrum fuer Informatik.

\bibitem[CJLW22]{cheng2018deterministic}
Kuan Cheng, Zhengzhong Jin, Xin Li, and Ke~Wu.
\newblock Deterministic document exchange protocols, and almost optimal binary
  codes for edit errors.
\newblock {\em Journal of the ACM (JACM)}, 69(6):1--39, 2022.
\newblock Preliminary version in 2018 IEEE 59th Annual Symposium on Foundations
  of Computer Science (FOCS).

\bibitem[CST22a]{9770830}
Roni Con, Amir Shpilka, and Itzhak Tamo.
\newblock Explicit and efficient constructions of linear codes against
  adversarial insertions and deletions.
\newblock {\em IEEE Transactions on Information Theory}, 68(10):6516--6526,
  2022.

\bibitem[CST22b]{9834672}
Roni Con, Amir Shpilka, and Itzhak Tamo.
\newblock Reed solomon codes against adversarial insertions and deletions.
\newblock In {\em 2022 IEEE International Symposium on Information Theory
  (ISIT)}, pages 2940--2945, 2022.

\bibitem[GHL22]{9954058}
Venkatesan Guruswami, Xiaoyu He, and Ray Li.
\newblock The zero-rate threshold for adversarial bit-deletions is less than
  1/2.
\newblock {\em IEEE Transactions on Information Theory}, pages 1--1, 2022.

\bibitem[GHS20]{GHS-stoc20}
Venkatesan Guruswami, Bernhard Haeupler, and Amirbehshad Shahrasbi.
\newblock Optimally resilient codes for list-decoding from insertions and
  deletions.
\newblock In {\em Proccedings of the 52nd Annual {ACM} Symposium on Theory of
  Computing}, pages 524--537, 2020.

\bibitem[GS99]{GuruswamiS99}
V.~Guruswami and M.~Sudan.
\newblock Improved decoding of reed-solomon and algebraic-geometry codes.
\newblock {\em IEEE Transactions on Information Theory}, 45(6):1757--1767,
  1999.

\bibitem[GW17]{GW}
Venkatesan Guruswami and Carol Wang.
\newblock Deletion codes in the high-noise and high-rate regimes.
\newblock {\em {IEEE} Trans. Information Theory}, 63(4):1961--1970, 2017.

\bibitem[Hae19]{Haeupler19}
Bernhard Haeupler.
\newblock Optimal document exchange and new codes for insertions and deletions.
\newblock In {\em 60th {IEEE} Annual Symposium on Foundations of Computer
  Science}, pages 334--347, 2019.

\bibitem[HRS19]{HRS19LinearTimeInsdel}
Bernhard Haeupler, Aviad Rubinstein, and Amirbehshad Shahrasbi.
\newblock {Near-Linear Time Insertion-Deletion Codes and (1+eps)-Approximating
  Edit Distance via Indexing}.
\newblock {\em {Proceeding of the ACM Symposium on Theory of Computing
  (STOC)}}, pages 697--708, 2019.

\bibitem[HS21]{HaeuplerS17}
Bernhard Haeupler and Amirbehshad Shahrasbi.
\newblock Synchronization strings: codes for insertions and deletions
  approaching the singleton bound.
\newblock {\em Journal of the ACM (JACM)}, 68(5):1--39, 2021.
\newblock Preliminary version in 49th Annual ACM SIGACT Symposium on Theory of
  Computing (STOC).

\bibitem[HSS18]{HSS18ListInsdel}
Bernhard Haeupler, Amirbehshad Shahrasbi, and Madhu Sudan.
\newblock Synchronization strings: List decoding for insertions and deletions.
\newblock {\em Proceeding of the International Colloquium on Automata,
  Languages and Programming (ICALP)}, pages 76:1--76:14, 2018.

\bibitem[HSV18]{SychnStringsInteractive}
Bernhard Haeupler, Amirbehshad Shahrasbi, and Ellen Vitercik.
\newblock Synchronization strings: Channel simulations and interactive coding
  for insertions and deletions.
\newblock {\em Proceeding of the International Colloquium on Automata,
  Languages and Programming (ICALP)}, pages 75:1--75:14, 2018.

\bibitem[HY18]{hayashi2018list}
Tomohiro Hayashi and Kenji Yasunaga.
\newblock On the list decodability of insertions and deletions.
\newblock In {\em 2018 IEEE International Symposium on Information Theory
  (ISIT)}, pages 86--90. IEEE, 2018.

\bibitem[SB19]{SimaB19}
Jin Sima and Jehoshua Bruck.
\newblock Optimal k-deletion correcting codes.
\newblock In {\em {IEEE} International Symposium on Information Theory}, pages
  847--851, 2019.

\bibitem[SZ99]{SZ99}
Leonard~J. Schulman and David Zuckerman.
\newblock Asymptotically good codes correcting insertions, deletions, and
  transpositions.
\newblock {\em {IEEE} Trans. Inf. Theory}, 45(7):2552--2557, 1999.
\newblock Preliminary version in SODA 1997.

\bibitem[Tho83]{1056765}
C.~Thommesen.
\newblock The existence of binary linear concatenated codes with reed - solomon
  outer codes which asymptotically meet the gilbert- varshamov bound.
\newblock {\em IEEE Transactions on Information Theory}, 29(6):850--853, 1983.

\bibitem[WZ17]{wachter2017list}
Antonia Wachter-Zeh.
\newblock List decoding of insertions and deletions.
\newblock {\em IEEE Transactions on Information Theory}, 64(9):6297--6304,
  2017.

\bibitem[YGM17]{Olgica17}
S.~M. Hossein~Tabatabaei Yazdi, Ryan Gabrys, and Olgica Milenkovic.
\newblock Portable and error-free dna-based data storage.
\newblock {\em Scientific Reports}, 7:2045--2322, 2017.

\end{thebibliography}
